\documentclass[lettersize,journal]{IEEEtran}
\usepackage{xcolor}
\usepackage{amsmath,amsfonts,amsthm}
\usepackage{amssymb}
\usepackage{mathtools}
\usepackage{algorithmic}
\usepackage{algorithm}
\usepackage{array}
\usepackage[caption=false,font=normalsize,labelfont=sf,textfont=sf]{subfig}
\usepackage{textcomp}
\usepackage{stfloats}
\usepackage{url}
\usepackage{verbatim}
\usepackage{graphicx}
\usepackage{cite}
\usepackage{multirow, multicol}
\usepackage[normalem]{ulem}
\usepackage{soul}
\usepackage{hyperref}

{}
{}
\newtheorem{remark}{Remark}{}
{}
\newtheorem{lemma}{Lemma}[]
\newtheorem{proposition}{Proposition}{}
\newtheorem{theorem}{Theorem}{}
\DeclareMathOperator*{\argmin}{argmin}

\newcommand{\x}{\mathsf{x}}
\newcommand{\LL}{\mathcal{L}}

\begin{document}

\title{Density-Driven Optimal Control for Efficient and Collaborative Multi-Agent Non-Uniform Coverage}

\author{Sungjun Seo,~\IEEEmembership{Member,~IEEE,} and Kooktae Lee,~\IEEEmembership{Member,~IEEE}
        % <-this % stops a space
\thanks{S. Seo and K. Lee are with the Department of Mechanical Engineering, 
New Mexico Institute of Mining and Technology, Socorro, NM 87801, USA 
(e-mail: sungjun.seo@student.nmt.edu; kooktae.lee@nmt.edu).}

\thanks{This work was supported in part by the U.S. National Science Foundation 
(NSF) CAREER Award under Grant CMMI-DCSD-2145810.}

\thanks{This is the author's accepted manuscript (AAM) of the paper 
accepted for publication in IEEE Transactions on Systems, Man, and Cybernetics: Systems. 
The final published version is available at 
https://doi.org/10.1109/TSMC.2025.3622075. 
© 2025 IEEE.}
}

% The paper headers
% \markboth{IEEE TRANSACTIONS ON SYSTEMS, MAN, AND CYBERNETICS: SYSTEMS,~Vol.~, No.~, AUGUST~2025}%
% {Shell \MakeLowercase{\textit{et al.}}: A Sample Article Using IEEEtran.cls for IEEE Journals
% }

\IEEEpubid{%0000--0000/00\$00.00~\copyright~2021 IEEE
}
% Remember, if you use this, you must call \IEEEpubidadjcol in the second
% column for its text to clear the IEEEpubid mark.

\maketitle

\begin{abstract}

This paper addresses the critical challenge of non-uniform area coverage in multi-agent systems, where certain regions require higher priority based on mission objectives. Traditional uniform coverage methods are often inadequate in many area coverage scenarios, as they may overlook real-world complexities where uniform coverage is neither feasible nor desirable. Although several non-uniform coverage methods exist, they frequently lack guarantees of optimality or fail to consider essential real-world constraints such as agent dynamics, the number of agents, available operation time, and decentralized control.

To overcome these limitations, a novel control scheme called Density-Driven Optimal Control (D$^2$OC) is proposed. The key innovation of D$^2$OC lies in the integration of optimal transport theory with multi-agent control, allowing agents to dynamically adjust their coverage according to a reference density distribution that reflects the mission’s priorities. Optimality is guaranteed by solving an optimization problem that incorporates the aforementioned real-world constraints. The control law is derived from the Lagrangian associated with the objective cost, utilizing optimal transport for both linear and nonlinear systems. Guarantees for global optimality and the existence of the optimal control input for linear systems are provided through an analytic solution. Additionally, an efficient data-sharing algorithm for decentralized control among multiple agents is proposed.

To evaluate the efficacy of the proposed control scheme, various simulation results are presented to compare its performance with existing methods.

\end{abstract}


\section*{Code Availability}
The MATLAB implementation of the proposed Density-Driven Optimal Control (D2OC) 
framework is publicly available at: 
\url{https://github.com/kooktaelee/D2OC}.

\section{Introduction}
\IEEEPARstart{I}{n} recent years, multi-agent systems have gained widespread attention due to their capability to tackle complex and large-scale tasks that a single agent cannot efficiently manage. These systems, composed of multiple agents working collaboratively, have shown great promise in a variety of applications. From search-and-rescue missions in disaster-stricken areas to large-scale environmental monitoring and military reconnaissance, multi-agent systems are crucial for tasks that require coordinated control and dynamic resource allocation \cite{drew2021multi, kouzehgar2020multi, huang2022multirobot}. Their ability to distribute workload, provide redundancy, and adapt to dynamic environments makes them an invaluable tool for both civilian and military operations. One of the most pressing challenges in these applications is achieving non-uniform area coverage, where certain regions demand higher priority based on mission objectives, environmental factors, or threat levels. Addressing this challenge requires a shift from traditional coverage approaches to more sophisticated strategies that account for varying levels of importance across different regions.

Conventional coverage algorithms, such as uniform coverage path planning (CPP), typically focus on distributing agents evenly across a given area \cite{oksanen2009coverage, avellar2015multi}. While this approach is effective in some scenarios, it fails to address real-world complexities where uniform coverage is neither feasible nor desirable. In practice, agents are often constrained by factors such as limited communication ranges, restricted energy resources, and finite operation time. These constraints demand more advanced solutions that can dynamically adjust to the needs of the mission. In environments where certain regions hold higher priority or require more frequent exploration, non-uniform coverage strategies become essential. These strategies must ensure that agents concentrate their efforts on high-priority areas while still maintaining coverage of the overall domain. By incorporating adaptive techniques that respond to mission-specific requirements, non-uniform coverage approaches offer a more practical and efficient alternative to uniform CPP methods.

Recent research on non-uniform coverage has proposed several methods to address these challenges. While these methods have shown promise, they often fail to account for critical operational constraints such as agent dynamics and number, decentralized communication, and operation time. The limitations of these existing approaches will be discussed together with uniform coverage in more detail in the following literature survey.

\textbf{Literature Survey} — Various approaches have been proposed for multi-agent area coverage, which can be broadly categorized into the following four classes:

1) \textit{Multi-Robot Coverage Path Planning (MRCPP):} Traditional MRCPP algorithms divide the area into smaller subregions using methods like cell decomposition \cite{azpurua2018multi, chi2021reusable}. These methods focus on the uniform distribution of agents and aim at optimizing path efficiency for complete coverage. While effective in ensuring uniform area exploration, MRCPP approaches do not account for regions requiring prioritized attention, which is crucial for non-uniform missions. Additionally, workload balancing among heterogeneous agents remains a significant challenge, especially in time-sensitive applications like search-and-rescue \cite{li2019multi}.

2) \textit{Ergodic Control and Spectral Multiscale Coverage (SMC):} The ergodic control framework aims to match the temporal distribution of agents to a predefined probability density function (PDF) representing regions of interest for non-uniform coverage. SMC introduced by Mathew et al. \cite{mathew2009spectral}, leverages the Fourier series to quantify ergodicity. While the SMC approach has been applied to scenarios such as exploration and surveillance \cite{surana2012coverage}, its primary limitation lies in the time taken to achieve ergodicity. In practice, finite-time constraints often prevent agents from converging to the ideal ergodic state, making it less suitable for time-critical missions.

3) \textit{Heat Equation Driven Area Coverage (HEDAC):} This approach models the area as a potential field governed by the heat equation, where the agents move along gradients derived from the field. The HEDAC method has demonstrated success in autonomous spraying and inspection tasks \cite{ivic2022constrained, ivic2019autonomous}. However, the lack of direct integration between agent dynamics and the heat equation parameters complicates the control of agent motion, leading to suboptimal performance in dynamic environments.

4) \textit{Density-Driven Control (D$^2$C):} Rooted in optimal transport (OT) theory, D$^2$C dynamically adjusts agent trajectories to match the priority distribution of a target area. By solving an OT problem, agents are guided towards regions of higher importance based on a weighted density function. D$^2$C, first proposed by Kabir et al. \cite{kabir2020receding}, offers a promising solution for non-uniform coverage, as it balances agent distribution with dynamic task prioritization. Recent extensions \cite{kabir2021wildlife}, \cite{lee2022density}, and \cite{seo2023density} have explored wildlife monitoring, 
decentralized implementations, and collision avoidance. However, its reliance on heuristic path planning rather than direct control over agent dynamics limits its effectiveness in real-time scenarios where responsiveness is crucial.

\textbf{Contribution} --- In this paper, a new multi-agent control scheme, referred to as Density-Driven Optimal Control (D$^2$OC), is proposed. In any given mission, one may construct a reference density distribution indicating the relative importance or priority in the domain. The key idea behind D$^2$OC is then to drive multiple agents such that their time-averaged behavior becomes as close to a given reference density distribution as possible. This is fundamentally different from the existing uniform coverage methods, such as MRCPP. Although other non-uniform coverage approaches like SMC, HEDAC, and D$^2$C have been developed, either optimality is not guaranteed in their methods, or they are impractical due to the lack of physical constraint considerations in their plans.

The novelty of our approach lies in the explicit integration of optimal transport theory with multi-agent control, enabling dynamic coverage while taking into account agent dynamics and number, operation times, and communication range. This is the first time that optimal transport theory has been applied to non-uniform coverage problems in a way that guarantees both practical applicability and optimality under real-world constraints.

The advantages of our work can be summarized as follows:
1) We formulate an optimization problem based on OT theory, which allows the derivation of an optimal control law for both linear and nonlinear agent dynamics. This optimization guarantees optimality and ensures that the derived control input minimizes resource wastage, something not considered in most existing coverage schemes.
2) Computational efficiency – Our method analytically derives the optimal control input without the need for complex numerical solvers (e.g., nonlinear programming), making it computationally efficient compared to ergodic-based approaches that involve heavy numerical computations.
3) Decentralized control method – A more efficient data-sharing method is proposed for decentralized control, allowing agents to coordinate their efforts without requiring a global communication network. This is crucial in real-world missions where maintaining continuous communication can be challenging.
4) Energy flexible control – Unlike other methods that assume uniform energy capacity among agents, D$^2$OC accommodates different energy levels or operation times among agents, optimizing coverage even under these disparities.
5) Collaborative coverage with heterogeneous platforms – Due to the flexibility of our approach, D$^2$OC supports heterogeneous agents with different dynamics or capabilities, making it highly versatile for applications involving diverse agent platforms.

In summary, our work contributes a novel, non-uniform coverage control scheme based on OT theory that ensures optimality, practicality, and adaptability in multi-agent systems. It fills a critical gap in existing research by addressing limitations in agent coordination, energy constraints, and the need for decentralized, dynamic control.

%%%%%%%%%%%%%%%%%%%%%%%%%%%%%%%%%%%%%%%%%%%%%%%%%%%%%%%%%%
\section{Preliminary and Problem description}\label{sec: prelim and prob descrip.}
\noindent{Notations}: The sets of real numbers and natural numbers, respectively, are denoted by $\mathbb{R}$ and $\mathbb{N}$. Moreover, $\mathbb{N}_0:=\mathbb{N}\cup \{0\}$. The symbol $\mathbb{R}^n$ represents an $n$-dimensional real vector space. Furthermore, $\mathbb{R}^{n\times m}$ denotes the set of real matrices of dimension $n\times m$. The symbol $\mathbf{I}_n$ stands for an $n\times n$ identity matrix. A transpose of a matrix is denoted by the upper script $^\top$, and the Kronecker product is represented by $\otimes$. The symbol diag$(d_1,...,d_n)$ denotes a diagonal matrix of a dimension $n\times n$ with the diagonal entries $d_1,..., d_n$. The set of pairs $\{(p_i,w_i)\ |\ p_i \in \mathbb{R}^n\text{ and }w_i \in \mathbb{R}\}$ implies a discrete distribution, where $p_i$ and $w_i$, respectively, are the position and corresponding weight of the $i$-th element in the discrete distribution. Positive definteness of the matrix $R$ and positive semidefiniteness of the matrix $Q$ are denoted by $R\succ \mathbf{0}$ and $Q\succeq \mathbf{0}$, respectively.

\subsection{Illustrative Example for Non-Uniform Coverage}

\begin{figure}[!h]
    \centering
    \subfloat[]{
    \includegraphics[width=0.32 \linewidth]{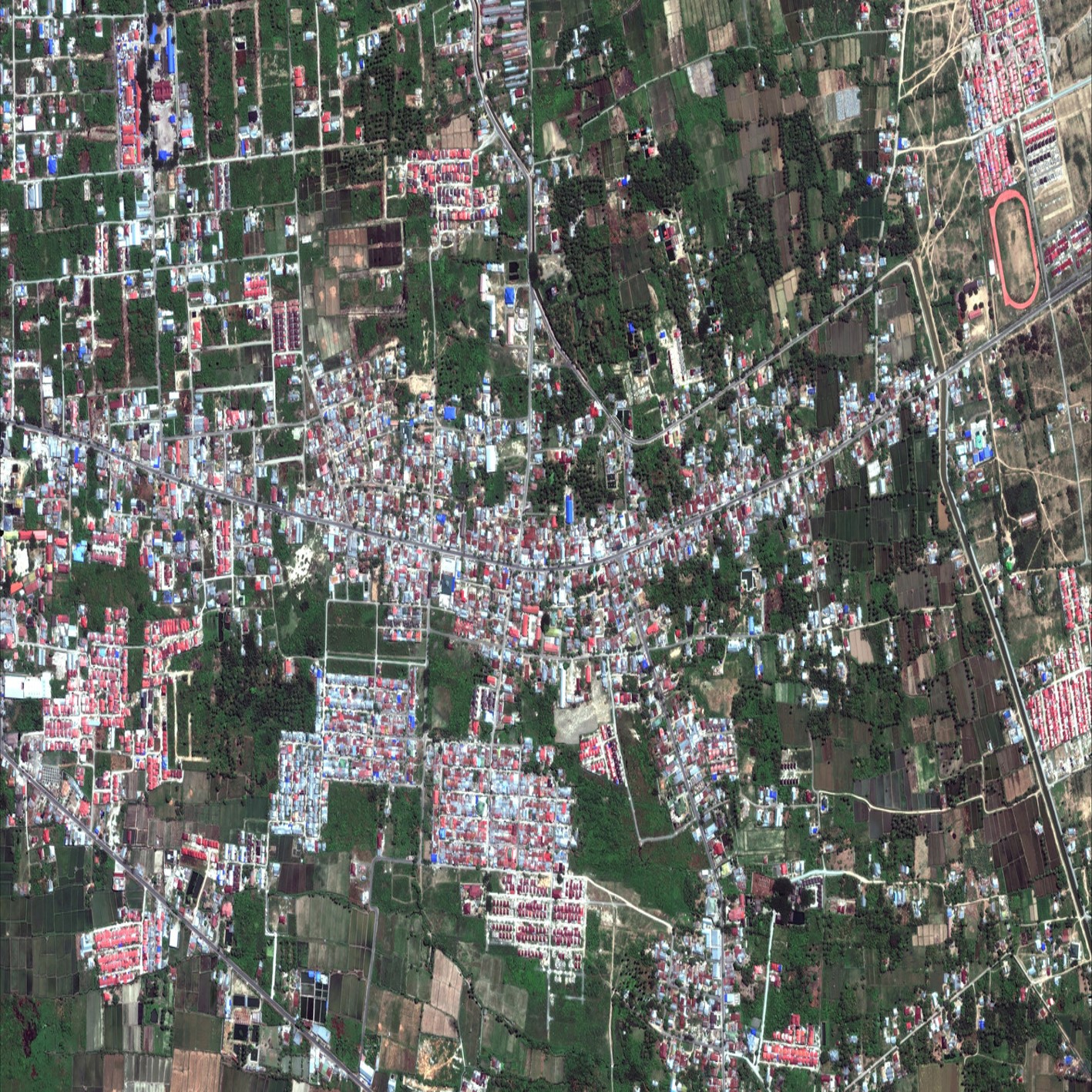}}
    \subfloat[]{
    \includegraphics[width=0.32 \linewidth]{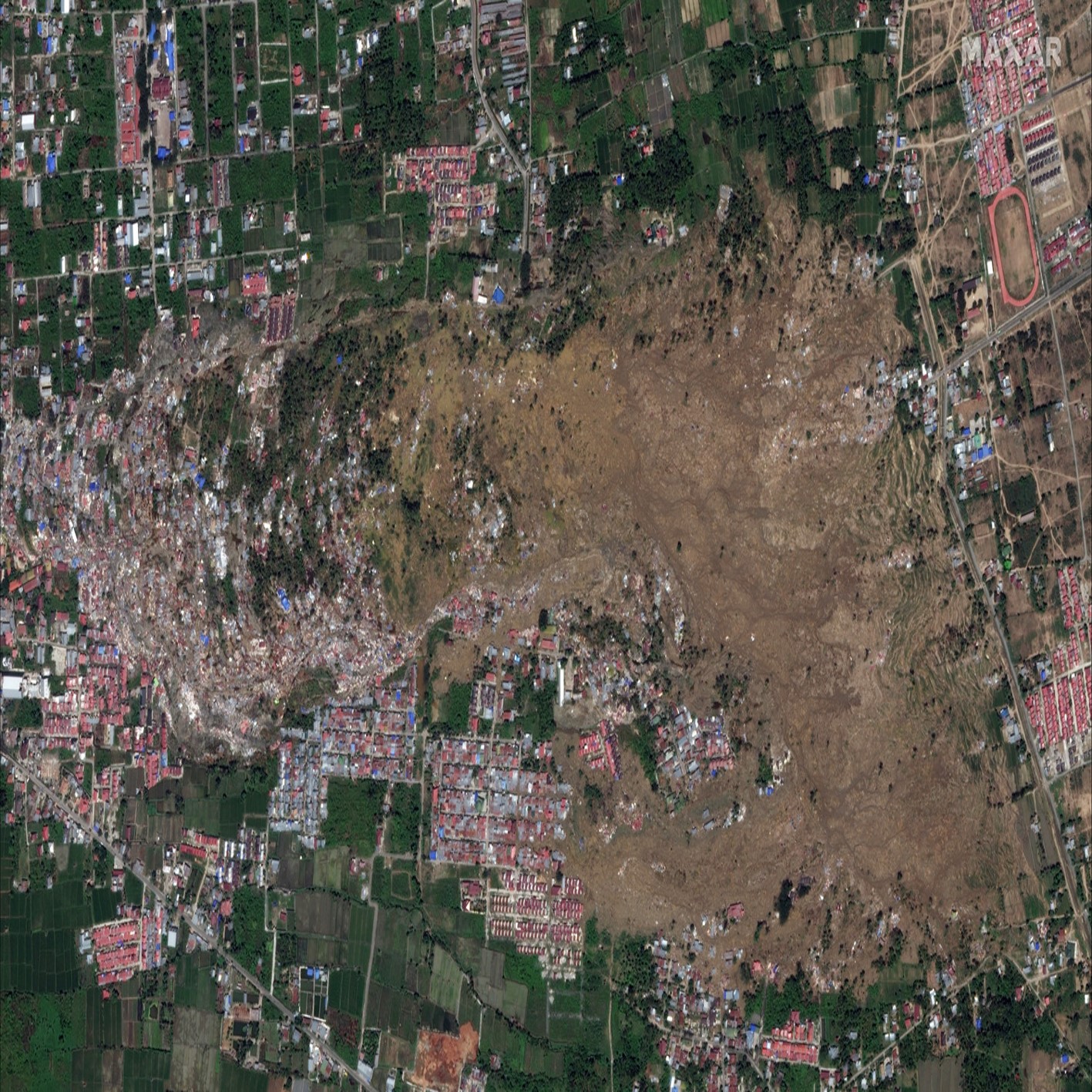}}
    \subfloat[]{
    \includegraphics[width=0.32\linewidth]{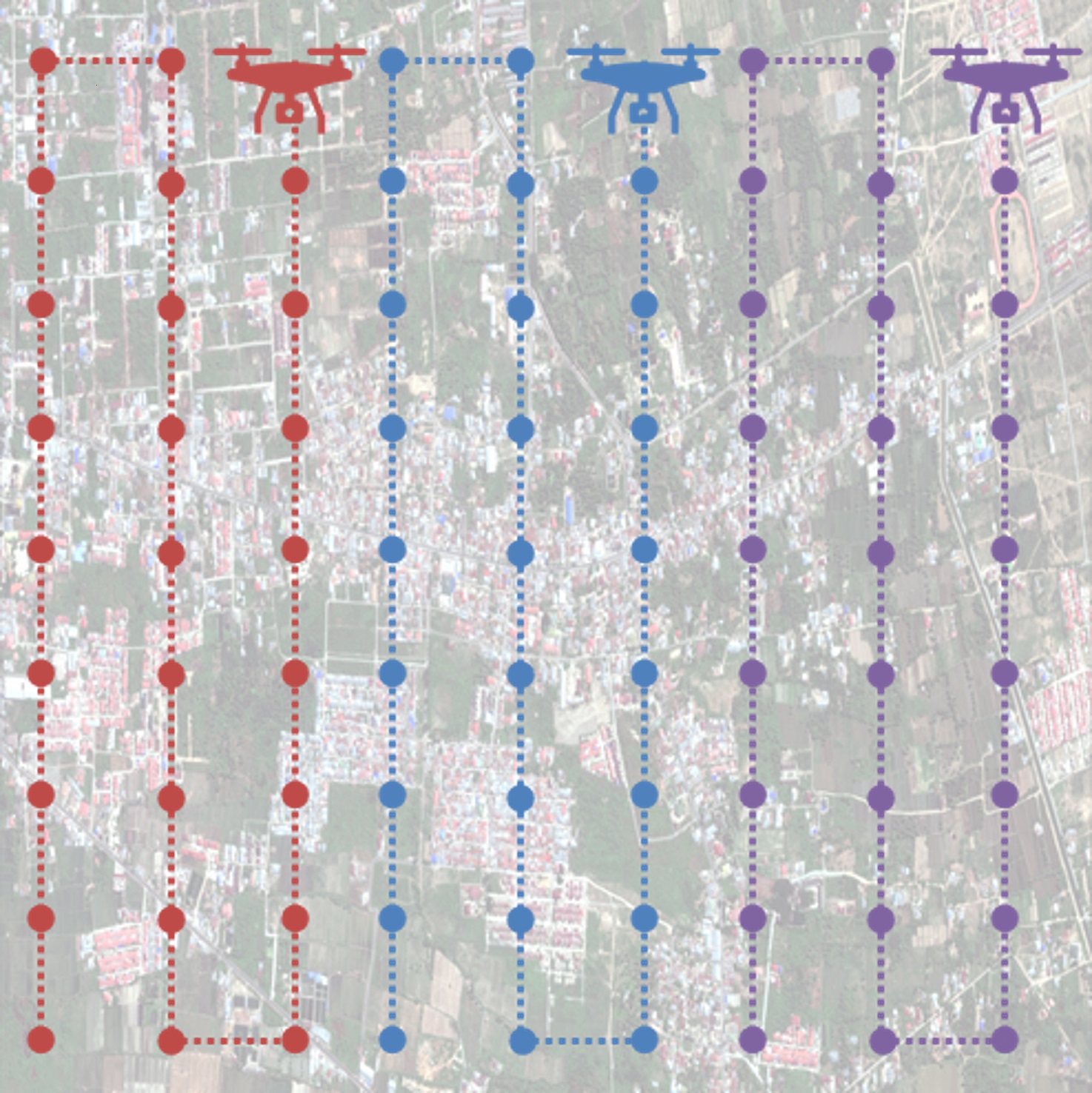}}
    \\\subfloat[]{
    \includegraphics[width=0.32\linewidth]{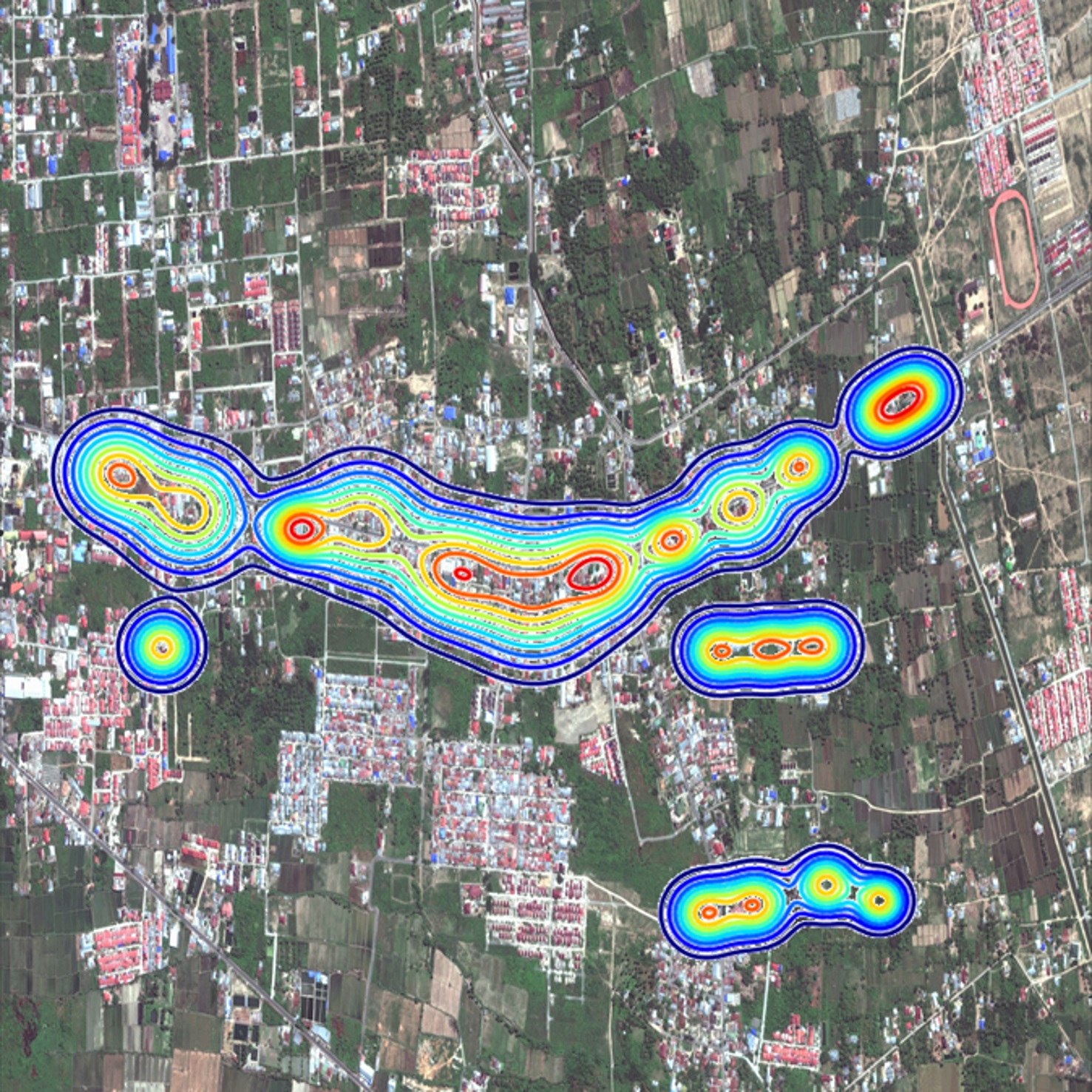}}
    \subfloat[]{
    \includegraphics[width=0.32\linewidth]{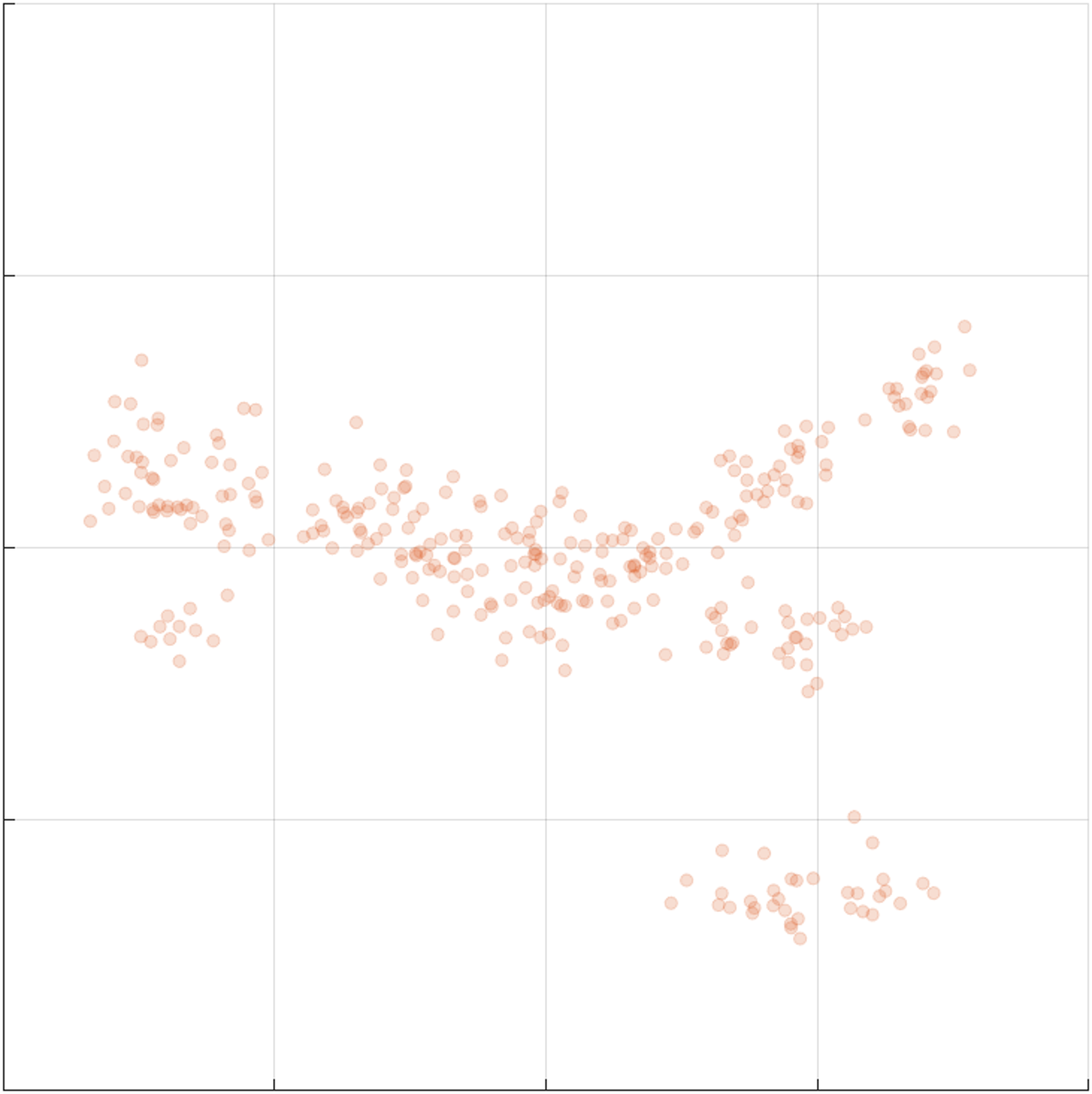}}
    \subfloat[]{
    \includegraphics[width=0.32\linewidth]{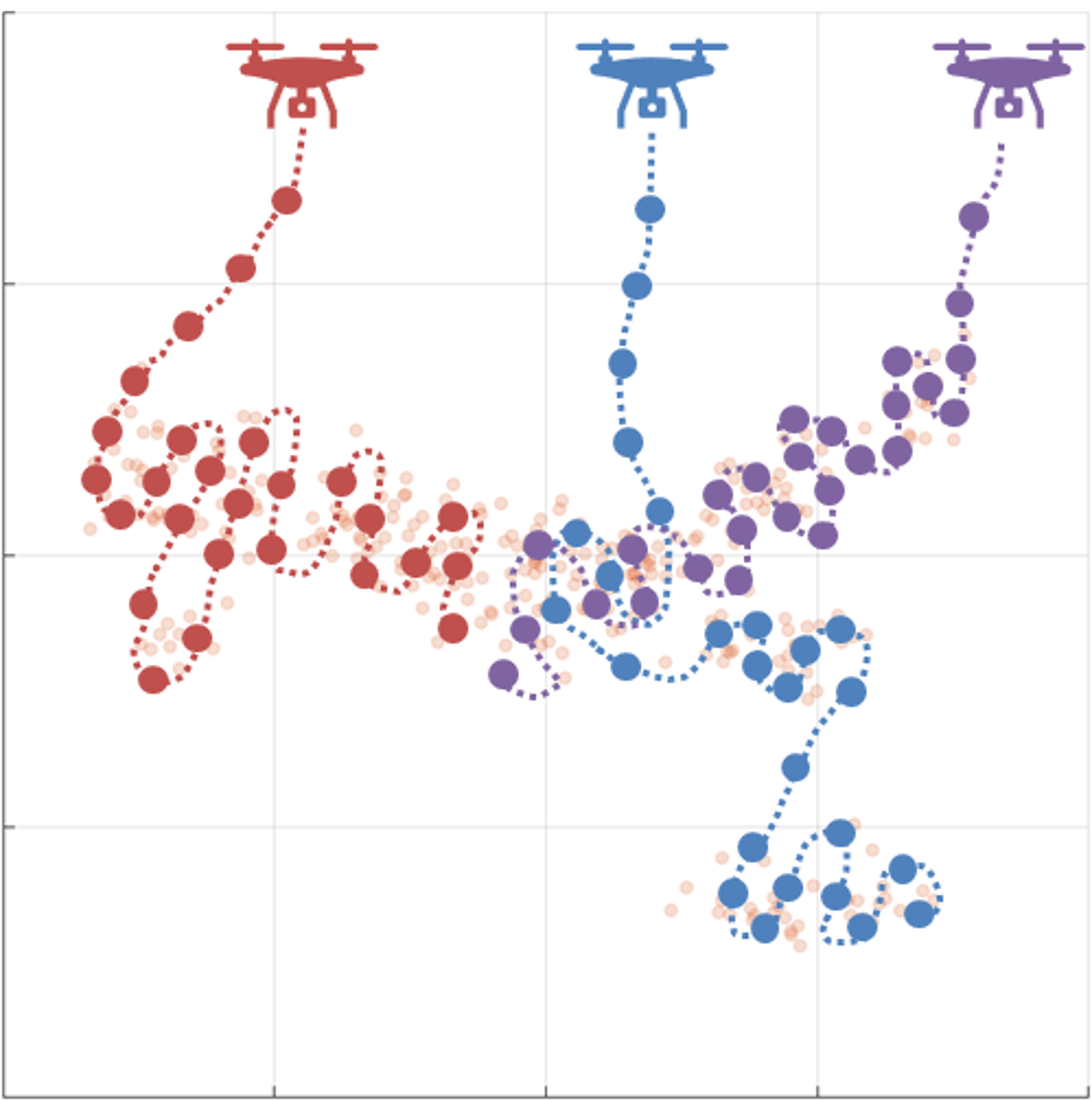}}
    \caption{Uniform vs. Non-Uniform Coverage: Satellite image of Petabo, Indonesia \cite{Nishan2018}. (a) before the tsunami; (b) after the tsunami; (c) uniform coverage path planning for a search and rescue mission; (d) a probability map of the victims' location based on pre-acquired information; (e) discretized point clouds generated from the probability map; (f) non-uniform coverage path utilizing the probability (or density)-based data.}
    \label{fig: Concept_Prob}
\end{figure}

Consider a large-scale search and rescue mission aimed at urgently locating victims after a catastrophic disaster, whether man-made or natural. As shown in Figs. \ref{fig: Concept_Prob}(a) and (b) from \cite{Nishan2018}, satellite images of Petabo, Palu City, Central Sulawesi, Indonesia, reveal the area before and after the 7.5-magnitude earthquake-triggered tsunami in 2018. In such scenarios, deploying an autonomous multi-agent system with onboard sensors is essential for locating survivors. One approach is uniform area coverage, as shown in Fig. \ref{fig: Concept_Prob}(c). However, given the urgency and limited resources---such as the number of agents, battery life, and communication range---uniform coverage may not maximize the rate of victim detection, leading to inefficiencies.

Rather than a uniform search, a probability-based approach—guided by factors such as disaster impact, residential locations, and population density—is more desirable, as seen in Fig. \ref{fig: Concept_Prob}(d). The probability (or density) map can be represented by point clouds (or sample points), as in Fig. \ref{fig: Concept_Prob}(e), facilitating an efficient search mission with a higher chance of detecting victims. A multi-agent system can then be deployed and driven by a control law such that the agents' discrete-time trajectories match the reference density map, as shown in Fig. \ref{fig: Concept_Prob}(f).

Given identical conditions (same multi-agent system and map), it is expected that non-uniform coverage (trajectories in Fig. \ref{fig: Concept_Prob}(f)) will yield a more desirable result—higher detection rates of victims—compared to uniform coverage (Fig. \ref{fig: Concept_Prob}(c)). In this context, this paper investigates a multi-agent non-uniform area coverage method for any given reference density map. This approach has potential applications in various domains, including environmental monitoring, smart farming, infrastructure inspection, and planetary exploration.

To solve this problem mathematically, we introduce Optimal Transport theory, which optimally allocates resources between two distributions, providing a control scheme that ensures the agents’ trajectories closely match the reference density map while considering physical constraints such as agent dynamics and number, operation time, and communication range limits.

\subsection{Optimal Transport Theory and Wasserstein Distance}
Optimal Transport (OT) is a research field studying optimal transportation or allocation of resources \cite{monge1781memoire, villani2008optimal}. It was first conceptualized to investigate the minimum cost of turning one pile into the other. Alternatively, it can be interpreted as measuring a distance between two probability density functions (PDFs) similar to measuring a distance between two points, such as a Euclidean distance.

 Given two discrete probability distributions, $\{(y_i, \alpha_{i})\}$ and $\{(q_j, \beta_{j})\}$, where 
 $i=1,...,M$ and $j=1,...,N$, the problem of minimizing the cost of transportation from $\{(q_j, \beta_{j})\}$ to $\{(y_i, \alpha_{i})\}$ in the OT theory is formulated by
 
\begin{equation}
\begin{aligned}
	\min_{\gamma_{ji}} \sum\nolimits_{j,i} \gamma_{ji}C_{ji}, \quad \gamma_{ji},\,C_{ji} \geq 0, \quad \forall i, j,
\end{aligned}\label{eqn: minimization problem in OT}
\end{equation}
where $\gamma_{ji}$ is the transportation plan from point $q_j$ to point $y_i$ and $C_{ji}$ is the transportation cost (e.g., distance) from point $q_j$ to point $y_i$.
The minimizer $\gamma^{*}$ is obtained as an optimal transportation plan that minimizes the objective function \eqref{eqn: minimization problem in OT}. 
This optimal transport problem is subject to the following constraints:
\begin{align}\label{eqn: OT constraint}
\sum_{i=1}^M\gamma_{ji} = \beta_{j}, \,\, \forall j,\,\sum_{j=1}^N\gamma_{ji} = \alpha_{i}, \,\,\forall i, \,\, \sum_{i=1}^{M}\alpha_i = \sum_{j=1}^{N}\beta_j = 1.
\end{align}

The first and second constraints are due to the weight of each point, and the last one is for mass conservation.
When the cost $C_{ji}$ is defined by the squared Euclidean distance between two points $y_i$ and $q_j$, the square root of the minimum cost is called the 2-Wasserstein distance, denoted by
\begin{align}
		\mathcal{W}_2 = \left(\min\nolimits_{\gamma_{ji}\geq0} \sum\nolimits_{i,j}\gamma_{ji} \lVert y_i - q_j \rVert^2\right)^{\frac{1}{2}},
\end{align} with the constraints in \eqref{eqn: OT constraint}.
The Wasserstein distance, as implied by its name, satisfies essential properties of distance metrics such as nonnegativity and triangle inequality. Furthermore, the 2-Wasserstein distance measures the dissimilarity between two different distributions.
In this sense, the 2-Wasserstein distance is employed as a tool to achieve density-driven control.

\subsection{Multi-Agent Coverage Strategy Using Wasserstein Distance}

For notational ease, points in the agents' trajectories and in the reference distribution are hereafter referred to as \textit{agent-points} and \textit{sample-points}, respectively. 
In this study, we aim to develop the optimal control scheme that drives the multiple agents so that the distribution of agent-points will closely match the distribution of sample-points.
The Wasserstein distance is employed to achieve the alignment between the agent-point distribution and the sample-point distribution. 
 Let the indices be defined as $r\in\mathbb{N}$ for the agent, $j\in\mathbb{N}$ for the sample-point, and $k\in\mathbb{N}$ for the discrete-time step. The agent-point distribution and the sample-point distribution can then be represented by $\{({}^{r}y^{k},{}^{r}\alpha^{k})\ |\ r=1,..., L \text{ and }k=1,...,{}^rM\}$ and $\{(q_j,\beta_j^0)\ |\ j=1,..., N\}$, respectively. We provide more details about each symbol as follows. Here, the symbol ${}^ry^k$ represents the spatial coordinate (or location) of agent $r$ at time $k$, and ${}^r\alpha^k$ denotes its corresponding weight. Similarly, $q_j$ represents the spatial coordinates of the $j$-th sample-point, and $\beta^0_j$ is its corresponding weight. The symbol ${}^rM$ refers to the total number of agent-points for agent $r$ calculated from the available operation time of the given agent platform divided by the sampling time.
 
 While the sample-point distribution is predefined, as information about sample-points is given in advance, the agent-point distribution has not yet been established. The spatial coordinate of the agent-point (or the position of the agent), ${}^{r}y^{k}$, evolves from its initial location, ${}^{r}y^{0}$, over its operation time, ${}^r M$, based on the agent's control inputs and dynamics. Corresponding weights of the agent-point, ${}^r\alpha^k$, can be assigned based on energy consumption between consecutive time intervals. 
 For example, a uniform weight can be considered as ${}^r\alpha^k=\frac{1}{\sum_r{}^{r}M}$ under the assumption that the energy (or available operation time) of the agent decreases linearly over discrete time. A more advanced energy consumption model (i.e., nonlinear) is also directly implementable to account for the time-varying consumption through the non-uniform weight.
 These weights must satisfy the last constraint in \eqref{eqn: OT constraint}.
 
 Finding the optimal control inputs that minimize the Wasserstein distance between the agent-point distribution $\{({}^{r}y^{k},{}^{r}\alpha^{k})\}$ and the sample-point distribution $\{(q_j,\beta_j^0)\}$ is computationally burdensome due to the large numbers both in sample- and agent-points, as well as the complex constraints requiring a numerical solver. To avoid this, we propose that, at each time step, local optimal transport occurs between a subset of sample-points and the current agent-point. In this way, the agents can focus on minimizing a local Wasserstein distance, with the cumulative effect over time leading to the desired outcome for density-driven control.

 In this OT framework, the weight of the sample-points, $\beta_j^k$, depends on the time step $k$. This reflects the coverage progress, as indicated by the decrease in $\beta_j^k$ when agents explore areas for non-uniform coverage. Specifically, the weights of sample-points near areas that an agent has covered will decrease according to a predefined rule. It is important to note that the positions of the sample-points remain stationary and are the same for all agents (i.e., all agents initially have identical information about the locations of sample-points).

Since the proposed D$^2$OC scheme is developed for decentralized control, each agent may progress differently in terms of area coverage, leading to variations in the sample-point weights in their memory unless they share and update this information. To account for this, the weight of each sample-point is indexed by the agent, denoted as ${}^{r}\beta_j^{k}$.

The optimization problem is then formulated to minimize the local Wasserstein distance between a subset of sample-points and the current agent-point. The detailed process and results are explained in the following section.

\section{Density-Driven Optimal Control (D$^2$OC)}\label{sec: D^2OC}
The proposed D$^2$OC scheme consists of three stages: Stage A - Optimal control, Stage B - Weight update, and Stage C - Weight sharing. 

In Stage A, the optimal control input is computed using the local sample-points (as depicted in Fig. \ref{fig: D2OC_schematic}(a)) chosen by the predefined priority index. After the agent moves to a new position by applying the optimal control, the weight of the sample-points near the agent is deducted in Stage B. Stage B is to update the area coverage progress by reducing the weights of sample-points in the areas swept by the agent. The deducted weights on sample-points will be transported to the agent as illustrated in Fig. \ref{fig: D2OC_schematic}(b).
Stages A and B are executed by each agent for fully decentralized control, which does not require global information such as other agents' states. For a collaborative multi-agent behavior, the agent shares the weight information in Stage C with other nearby agents within the communication range. In each time step, these stages are sequentially applied to the agents and repeated until all the weights of the sample-points are transported to the agent-points. In the following subsections, each stage will be discussed in detail.
\begin{figure}
    \centering
    \includegraphics[width=0.8\linewidth]{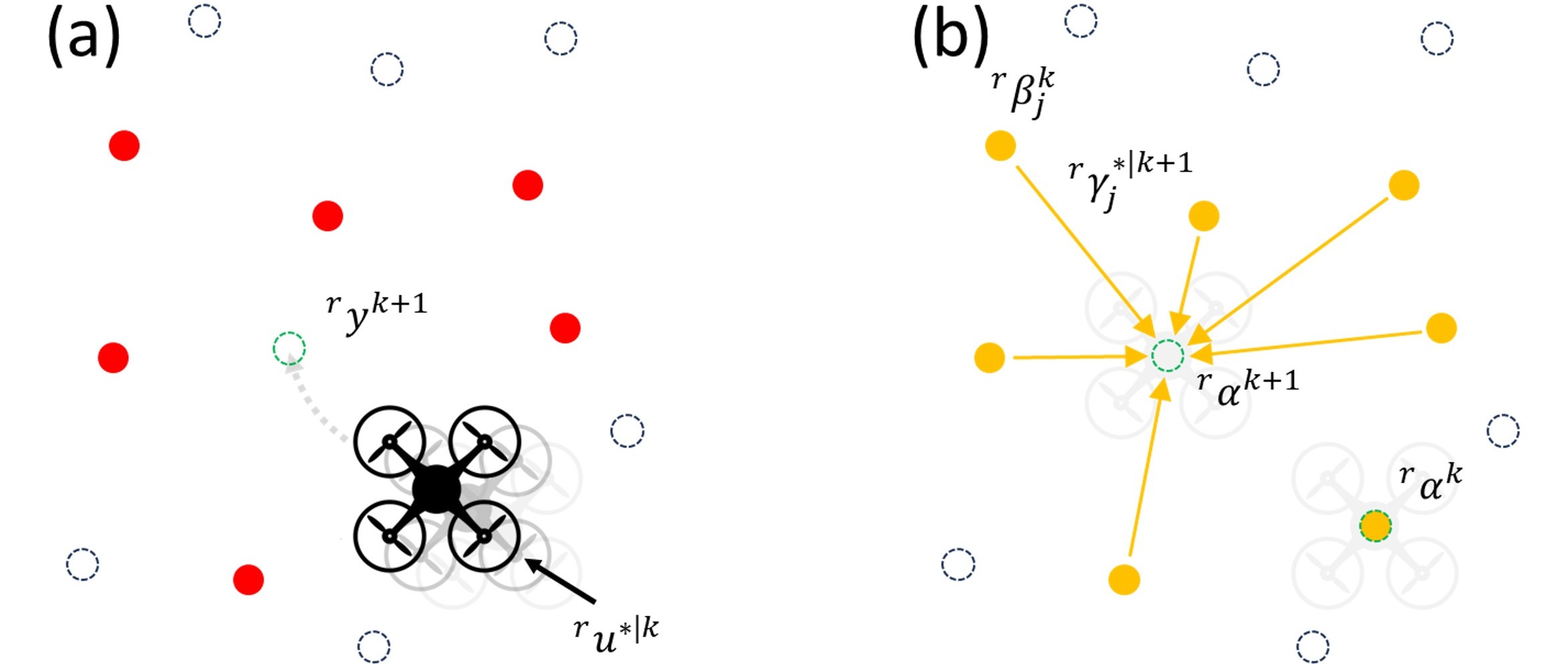}
    \caption{Schematic drawing of each stage in the D$^2$OC scheme: (a) Stage A - Optimal control stage: Local sample-points (red circles) are chosen from sample-points (hollow black circles). Optimal control input ${}^{r}u^{*|k}$ is applied and the agent moves to the new location ${}^{r}y^{k+1}$; (b) Stage B - Weight update stage: The optimal transportation ${}^{r}\gamma^{*|k+1}_j$ taken out of the remaining weight ${}^{r}\beta^k_j$ is transported to the newly created agent-point (hollow green circles).}
    \label{fig: D2OC_schematic}
\end{figure}
\subsection{Optimal Control Stage}
The optimal control stage is further divided into two steps: the local sample-points selection and the optimal control step. In the local sample-points selection step, the agent selects some sample-points based on the weight-normalized Euclidean (wnE) distance as follows:
\begin{align}
     d_{wnE}(j,{}^{r}y^k) = \frac{|| q_{j}-{}^{r}{y}^{k}||}{{{}^{r}\beta^k_{j}}},\label{eqn: weight-normalized distance}
\end{align}
where ${}^{r}\beta^k_j$ is the remaining weight of the sample-point $q_j$ at time $k$ and ${}^{r}y^k$ is the location of agent $r$. Based on this metric, sample-points that are closer to the agent and have a higher weight are chosen as local sample-points.
To specify how many sample-points should be selected for the local sample-points, we propose that the total amount of weights in the local sample-points should be the same as the weight of the single agent-point to satisfy the mass conservation law of the OT theory. In each time step, the weight of the single agent-point, ${}^{r}\alpha^k$, is calculated by the consumed energy in the given discrete-time interval divided by the initially available energy. 
It is worth noting that the energy flexibility of D$^2$OC is achievable through this idea -- the incorporation of the initially available energy, which may vary from agent to agent depending on the given agent platform.

The set of local sample-points at time $k$ is denoted by ${}^{r}S^k$, which will be used to derive optimal control for D$^2$OC. 
To clearly show the process and result, the linear and nonlinear cases are separated. Since this stage is identically applied across all agents in a decentralized control setup, the agent index $r$ is omitted except for the case where the meaning is unclear.
\subsubsection{Linear Time-Invariant system case}
Consider a discrete-time Linear Time-Invariant (LTI) system given by 
\begin{equation}
    \begin{aligned}
    \mathsf{x}^{k+1} = A\mathsf{x}^{k}+Bu^{k},\ y^{k} = C\mathsf{x}^{k}, 
\end{aligned}\label{eqn: LTI system}
\end{equation}
where $\mathsf{x}\in \mathbb{R}^n$ and $u\in \mathbb{R}^m$ are a state vector and control input of an agent in the multi-agent system, respectively, and $y$ is the position of the agent in the Cartesian coordinate space. The proposed optimal control input can be applied to each agent in a decentralized control fashion. To align the trajectories of the agent closely with the reference distribution, the optimization problem is formulated by utilizing the Wasserstein distance associated with the local sample-points with a given horizon length $T$ as follows:
\begin{equation}
\begin{aligned}
	&\min_{u} J = \Phi(\mathsf{x}^{k+T})
 \\&+\sum_{i=k}^{k+T-1}\left(\dfrac{1}{2}(\mathcal{W}^{i|k})^2 + \dfrac{1}{2}(\mathsf{x}^{i})^{\top}Q\mathsf{x}^{i} +  \dfrac{1}{2}(u^{i})^{\top}Ru^{i}\right)
  \\
 &\text{subject to} \quad \mathsf{x}^{k+1} = A\mathsf{x}^k + Bu^k,
\end{aligned}\label{eqn: minimization problem}
\end{equation}
\vspace{-0.025in}
{\allowdisplaybreaks
\begin{align*}
&\begin{aligned}
\text{where}     &\ \  Q\succeq \mathbf{0},\ R\succ \mathbf{0},\ (\mathcal{W}^{i|k})^2= \min_{\gamma_{j}\geq0} \sum_{j \in \mathcal{S}^{k}} \gamma_j ||C\mathsf{x}^{i}-q_j||^{2},
     \\&\ \  \Phi(\mathsf{x}^{k+T})=\dfrac{1}{2}(\mathcal{W}^{k+T|k})^2 + \dfrac{1}{2}(\mathsf{x}^{k+T})^{\top}Q\mathsf{x}^{k+T}.
 \end{aligned}
 \end{align*}}
In the cost function of \eqref{eqn: minimization problem}, $\Phi(\x^{k+T})$ is the terminal cost over the given horizon, and $\mathcal{W}^{i|k}$, referred to as the \textit{local Wasserstein distance}, is the Wasserstein distance between the agent-point at index $i=k,\cdots, k+T-1$ and the local sample-points at time $k$. Notice that index $i$ is used to denote the multiple agent-points in the future, given the horizon length $T$. The third and fourth terms of the cost function, respectively, are to penalize the input and state vectors with a positive semidefinite matrix $Q$ and a positive definite matrix $R$.
Note that Stage A is implemented by each agent for fully decentralized control, and hence no global information is required in \eqref{eqn: minimization problem}.

To obtain the optimal control input for the D$^2$OC scheme, the following proposition is introduced first.
\begin{proposition}\label{prop: invertibility}
    Let the symmetric matrix E be defined by
    \begin{align}
        E =  \begin{bmatrix}
        E_{11} & E_{12} & \mathbf{0}_{} \\
        E_{12}^{\top} & \mathbf{0} & E_{23}\\
        \mathbf{0} & E_{23}^{\top} & E_{33}
        \end{bmatrix},\label{eqn: E}
    \end{align}
    where $E_{ij}$ is a real submatrix, $E_{11}\in \mathbb{R}^{n\times n}$ and $E_{33}\in \mathbb{R}^{m\times m}$ are square matrices, and $E_{12} \in \mathbb{R}^{n\times n}$ is invertible. Then, the matrix $E$ is invertible if $E_{11}$ is symmetric positive semidefinite and $E_{33}$ is positive definite. 
    
    Furthermore, the inverse of the matrix E is calculated by
    \begin{align}
        E^{-1} =  \begin{bmatrix}
        (E^{-1})_{11} & (E^{-1})_{12} & (E^{-1})_{13} \\
        (E^{-1})_{12}^{\top} & (E^{-1})_{22} & (E^{-1})_{23}\\
        (E^{-1})_{13}^{\top} & (E^{-1})_{23}^{\top} & (E^{-1})_{33}
        \end{bmatrix},\label{eqn: inverse of E}
    \end{align}
    where
    {\allowdisplaybreaks
    \begin{align*}
        (E^{-1})_{11}&={(E_{12}^{-1})}^{\top}E_{23}(E^{-1})_{33}E_{23}^\top E_{12}^{-1},
        \\(E^{-1})_{12}&=(E_{12}^{-1})^\top
        \\\phantom{(E^{-1})_{12}}&\phantom{=\ }-{(E_{12}^{-1})^\top}E_{23}(E^{-1})_{33}E_{23}^{\top}E_{12}^{-1}E_{11}{(E_{12}^{-1})^\top},
        \\(E^{-1})_{13}&=-{(E_{12}^{-1})^\top}E_{23}(E^{-1})_{33},
        \\(E^{-1})_{22}&=E_{12}^{-1}E_{11}{(E_{12}^{-1})^\top}
        \\\phantom{(E^{-1})_{22}}&\phantom{=\ }(E_{23}(E^{-1})_{33}E_{23}^\top E_{12}^{-1}E_{11}{(E_{12}^{-1})^\top}-\mathbf{I}_n),
        \\(E^{-1})_{23}&=E_{12}^{-1}E_{11}{(E_{12}^{-1})^\top}E_{23}(E^{-1})_{33},
        \\(E^{-1})_{33}&=(E_{33}+E_{23}^\top E_{12}^{-1}E_{11}{(E_{12}^{-1})^\top}E_{23})^{-1}.
    \end{align*}}
\end{proposition}
\begin{proof}
    The reduced row-echelon form (RREF) of the matrix $E$ is given by
    \begin{align}
        E \xrightarrow[]{\text{RREF}}
        \begin{bmatrix}
        \mathbf{I}_n & \mathbf{0} & (E_{12}^{\top})^{-1}E_{23}\\
        \mathbf{0} & \mathbf{I}_n & -E_{12}^{-1}E_{11}(E_{12}^{\top})^{-1}E_{23} \\
        \mathbf{0} & \mathbf{0} & E_{33}+E_{23}^{\top}E_{12}^{-1}E_{11}(E_{12}^{\top})^{-1}E_{23}
        \end{bmatrix}.\nonumber
        \end{align}
        Since the rank of the matrix remains the same even after the row reduction process, the matrix $E$ has full rank if and only if the submatrix $E_{33}+\{(E_{12}^{-1})^{\top}E_{23})\}^{\top}E_{11}\{(E_{12}^{-1})^{\top}E_{23}\}$ has full rank, which is the case if $E_{11}$ is symmetric positive semidefinite and $E_{33}$ is positive definite. 
        
        The inverse of the matrix $E$ is obtained by the row reduction process. 
        The detailed derivation process is omitted due to the limited space, although it is straightforward.
\end{proof}
Based on Proposition \ref{prop: invertibility}, the optimal control input for D$^2$OC of the LTI system is derived as follows.

\begin{theorem}\label{Theorem: LTI D2OC Optimal control}
    For any LTI system having the structure in \eqref{eqn: LTI system}, consider the optimal control problem \eqref{eqn: minimization problem} to achieve D$^2$OC.
    The optimal control input $u^{*|k}$ that minimizes the cost function \eqref{eqn: minimization problem} is obtained by
    \begin{align}
    \begin{aligned}
        &u^{*|k} = \begin{bmatrix}
        \mathbf{I}_m & \mathbf{0} & \cdots & \mathbf{0}
        \end{bmatrix}\bar{u}^k \text{ with }\\
        &\bar{u}^k = (E^{-1})_{13}^{\top}F_{1} + (E^{-1})_{23}^{\top}F_{2},
        \end{aligned}
        \label{eqn: optimal_u}
    \end{align}
    where $(E^{-1})_{13}^{\top}$ and $(E^{-1})_{23}^{\top}$ denote the submatrices of the inverse matrix $E^{-1}$, as defined in \eqref{eqn: inverse of E}, while the block matrix $E$ in the form of \eqref{eqn: E} consists of the following submatrices:
    {\allowdisplaybreaks
    \begin{align}
    &E_{11}=\bar{Q} \otimes \mathbf{I}_{T}, \ E_{23} = E_{32}^{\top} = B \otimes \mathbf{I}_{T}, \  E_{33} = R \otimes \mathbf{I}_{T}, \nonumber 
    \\&\bar{Q}=\{(\sum_{j \in \mathcal{S}^{k}} \gamma_j)C^{\top}C+ Q\},\nonumber
    \\&E_{12} = E_{21}^{\top} = 
    \begin{bmatrix}
    -\mathbf{I}_n & A^{\top} & \mathbf{0}_{} & \cdots & \mathbf{0}_{}
    \\\mathbf{0}_{} & -\mathbf{I}_n & A^{\top} & \ddots & \vdots
    \\\vdots & \ddots & \ddots & \ddots & \mathbf{0}_{}
    \\\vdots  &   & \ddots & -\mathbf{I}_n & A^{\top} 
    \\\mathbf{0}_{} & \cdots  & \cdots  & \mathbf{0}_{}& -\mathbf{I}_n
    \end{bmatrix}. \label{eqn: E and F}
    \end{align}}
    
    Moreover, $F_1$ and $F_2$ in \eqref{eqn: optimal_u} are defined by
    
    \begin{align*}
    & F_1 = 
    \begin{bmatrix}
    (\sum_{j \in \mathcal{S}^{k}} \gamma_j)C^{\top}\overline{q}^{k}\\
    \vdots\\
    (\sum_{j \in \mathcal{S}^{k}} \gamma_j)C^{\top}\overline{q}^{k}
    \end{bmatrix}, \, \nonumber
     F_2 = 
    \begin{bmatrix}
    -A\mathsf{x}^k
    \\\mathbf{0}_{}
    \\\vdots
    \\\mathbf{0}_{}
    \end{bmatrix},
    \end{align*}
    where $\overline{q}^{k}$ is defined as $\overline{q}^{k}=(\sum_{j \in \mathcal{S}^{k}} \gamma_j q_{j})/(\sum_{j \in \mathcal{S}^{k}} \gamma_j)$.
\end{theorem}
\begin{proof}
    To derive the minimum points of the cost function $J$ in \eqref{eqn: minimization problem}, subject to the agent's dynamics in \eqref{eqn: LTI system}, the method of Lagrange multipliers is employed in this proof.
    The Lagrangian associated with \eqref{eqn: minimization problem} is formulated by
    \begin{equation}
    \begin{aligned}
         \mathcal{L} = &\Phi(\mathsf{x}^{k+T})+\sum_{i=k}^{k+T-1}\left(\dfrac{1}{2}(\mathcal{W}^{i|k})^2 + \dfrac{1}{2}(u^{i})^{\top}Ru^{i}+\right.\\
         &\left.\dfrac{1}{2}(\mathsf{x}^{i})^{\top}Q\mathsf{x}^{i}+(\lambda^{i+1})^{\top}(A\mathsf{x}^{i}+Bu^{i}-\mathsf{x}^{i+1})\right),
     \end{aligned}\label{eqn: Lagrangian}
    \end{equation}
    where $\lambda^{i+1}$ is the costate vector for the constraint function $A\mathsf{x}^{i}+Bu^{i}-\mathsf{x}^{i+1}$. The cost function $J$ should have stationary points at $\x^{*|i},\ \lambda^{*|i},\ u^{*|i}$ where the Lagrangian $\LL$ is stationary with respect to $\x^i$, $u^i$, and $\lambda^i$. This implies   $\partial \LL / \partial \x^{i+1}(\x^{*|k+1},\ldots,\allowbreak\x^{*|k+T},\allowbreak\lambda^{*|k+1},\ldots,\allowbreak \lambda^{*|k+T},u^{*|k},\ldots,\allowbreak u^{*|k+T-1}) = 0$, $\partial \LL / \partial \lambda^{i+1}(\x^{*|k+1},\ldots,\allowbreak\x^{*|k+T},\allowbreak\lambda^{*|k+1},\ldots,\allowbreak \lambda^{*|k+T},u^{*|k},\ldots,\allowbreak u^{*|k+T-1}) = 0$, and $\partial \LL / \partial u^{i}(\x^{*|k+1},\ldots,\allowbreak\x^{*|k+T},\allowbreak\lambda^{*|k+1},\ldots,\allowbreak\lambda^{*|k+T},u^{*|k},\ldots,\allowbreak u^{*|k+T-1}) = 0$, for $i=k,k+1,\cdots,k+T-1$.
    
    Thus, the necessary conditions for optimality and boundary conditions are given by
    {\allowdisplaybreaks
    \begin{subequations}\label{eqn: pontryagin}
    \begin{align}
    	&\dfrac{\partial \mathcal{L}}{\partial \mathsf{x}^{i}} = \mathbf{0} = 
        \dfrac{1}{2}\dfrac{\partial\mathcal(\mathcal{W}^{i|k})^2}{\partial \mathsf{x}^{i}} + Q\mathsf{x}^{*|i} + A^{\top}\lambda^{*|i+1}-\lambda^{*|i} \nonumber
        \\& \quad=\sum_{j \in \mathcal{S}^{k}} \gamma_j C^{\top}(C\mathsf{x}^{*|i}-q_j)+ Q\mathsf{x}^{*|i} + A^{\top}\lambda^{*|i+1}-\lambda^{*|i}\nonumber
        \\& \quad= \bar{Q}\mathsf{x}^{*|i}-(\sum_{j \in \mathcal{S}^{k}} \gamma_j)C^{\top}\overline{q}^{k} + A^{\top}\lambda^{*|i+1}-\lambda^{*|i}, \label{eqn: pontryagin-1}
        \\&\dfrac{\partial \mathcal{L}}{\partial \lambda^{i'+1}} = \mathbf{0} = A\mathsf{x}^{*|i'} + Bu^{*|i'} - \mathsf{x}^{*|i'+1},\label{eqn: pontryagin-2}
        \\&\dfrac{\partial \mathcal{L}}{\partial u^{i'}} =\mathbf{0}= Ru^{*|i'} + B^{\top}\lambda^{*|i'+1}, \label{eqn: pontryagin-3}
        \\
        &\qquad \left(
        \begin{aligned}
        &\text{for } i = k+1, k+2, \cdots ,k+T-1\nonumber
        \\&\text{and } i'=k,k+1,\cdots,k+T-1 \nonumber
        \end{aligned} \right)
        \\&\dfrac{\partial \mathcal{L}}{\partial \mathsf{x}^{k+T}} = \mathbf{0} = 
         \dfrac{1}{2}\dfrac{\partial\Phi(\mathsf{x}^{k+T})}{\partial \mathsf{x}^{k+T}} - \lambda^{*|k+T} \nonumber
         \\&= \bar{Q}\mathsf{x}^{*|k+T}-(\sum\nolimits_{j \in \mathcal{S}^{k}} \gamma_j)C^{\top}\overline{q}^{k} - \lambda^{*|k+T}.\label{eqn: pontryagin-4}
    \end{align}
    \end{subequations}}
    The equations \eqref{eqn: pontryagin-1} are for the costate equations, \eqref{eqn: pontryagin-2} are agent's dynamics, \eqref{eqn: pontryagin-3} are for the optimal inputs, and \eqref{eqn: pontryagin-4} are 
    boundary conditions at the final time $k+T$.
    These equations can be rearranged by 
    {\allowdisplaybreaks
    \begin{equation}
    \begin{aligned}\label{eqn: necessary conditions}
        \bar{Q}\mathsf{x}^{*|i} -\lambda^{*|i} + A^{\top}\lambda^{*|i+1} &= (\sum\nolimits_{j \in \mathcal{S}^{k}} \gamma_j)C^{\top}\overline{q}^{k},
        \\\bar{Q}\mathsf{x}^{*|k+T} -\lambda^{*|k+T} &= (\sum\nolimits_{j \in \mathcal{S}^{k}} \gamma_j)C^{\top}\overline{q}^{k},
        \\-\mathsf{x}^{*|k+1}+Bu^{*|k} &=-A\mathsf{x}^{k},
        \\A\mathsf{x}^{*|i}-\mathsf{x}^{*|i+1}+Bu^{*|i} &=\mathbf{0},
        \\Ru^{*|i} + B^{\top}\lambda^{*|i+1} &= \mathbf{0},
        \\(\text{for }i = &k+1, k+2, \cdots ,k+T-1.) 
    \end{aligned}
    \end{equation}
    
    Defining the augmented vectors $\bar{\mathsf{x}}^{k+1} =\allowbreak [(\mathsf{x}^{*|k+1})^{\top}, \,\allowbreak \ldots\, ,\allowbreak(\mathsf{x}^{*|k+T})^{\top}]^{\top}$, $\bar{\lambda}^{k+1} = \allowbreak[(\lambda^{*|k+1})^{\top}, \, \allowbreak\ldots, \,\allowbreak(\lambda^{*|k+T})^{\top}]^{\top}$, 
    and
    $\bar{u}^{k} = \allowbreak[(u^{*|k})^{\top}, \,\allowbreak\ldots, \,\allowbreak(u^{*|k+T-1})^{\top}]^{\top}$},
    the equations \eqref{eqn: necessary conditions} are rewritten in the matrix form as
    \begin{align}
        E\begin{bmatrix}
        \Bar{\mathsf{x}}^{k+1} \\
        \Bar{\lambda}^{k+1} \\
        \Bar{u}^{k}
        \end{bmatrix} =
        \begin{bmatrix}
        E_{11} & E_{12} & \mathbf{0}_{} \\
        E_{21} & \mathbf{0} & E_{23}\\
        \mathbf{0} & E_{32} & E_{33}
        \end{bmatrix} 
        \begin{bmatrix}
        \Bar{\mathsf{x}}^{k+1} \\
        \Bar{\lambda}^{k+1} \\
        \Bar{u}^k
        \end{bmatrix} =
        \begin{bmatrix}
        F_{1} \\
        F_{2} \\
        \mathbf{0}
        \end{bmatrix}, \label{eqn: matrix form}
    \end{align}
    where the submatrices of the block matrix $E$, namely $E_{11},\ E_{12},\ E_{21},\ E_{23},\ E_{32},\ E_{33}$, are defined in \eqref{eqn: E and F}.
    The matrix $E_{12}$ already has full rank in the row echelon form and is hence invertible. Further, $E_{11}$ is symmetric positive semidefinite and $E_{33}$ is positive definite due to the property of $Q$ and $R$, leading to the matrix $E \in \mathbb{R}^{(2nT+mT)\times(2nT+mT)}$ being invertible, based on Proposition \ref{prop: invertibility}. As a result,
    \begin{align}
    \begin{bmatrix}
    \Bar{\mathsf{x}}^{k+1} \\
    \Bar{\lambda}^{k+1} \\
    \Bar{u}^k
    \end{bmatrix} =
    \begin{bmatrix}
    E_{11} & E_{12} & \mathbf{0}_{} \\
    E_{21} & \mathbf{0} & E_{23}\\
    \mathbf{0} & E_{32} & E_{33}
    \end{bmatrix}^{-1} 
    \begin{bmatrix}
    F_{1} \\
    F_{2} \\
    \mathbf{0}
    \end{bmatrix}. \label{eqn: augmented variables}
    \end{align}
    
    By substituting $E^{-1}$ with the equation \eqref{eqn: inverse of E}, the analytic form of the augmented control input vector $\bar{u}^k$ is calculated by
    \begin{align}
        \bar{u}^k = (E^{-1})_{13}^{\top}F_{1} + (E^{-1})_{23}^{\top}F_{2}.\nonumber
    \end{align}
    
    Finally, the optimal control input for D$^2$OC at time $k$ is obtained as $u^{*|k} = \begin{bmatrix}
    \mathbf{I}_m & \mathbf{0} & \cdots & \mathbf{0}
    \end{bmatrix}
    \bar{u}^k$.
\end{proof}
\begin{remark}{\bf (Optimality Guarantee)}
Generally, the partial derivatives \eqref{eqn: pontryagin} are known to be the necessary condition for the optimality (minimum in this case), however, it is guaranteed that the control input \eqref{eqn: optimal_u} leads to the global minimum since the system is linear and the Hessian matrix $\partial^2 \mathcal{L}/(\partial \mathbf{u}^k)^2$, where $\mathbf{u}^k = \allowbreak[(u^{k})^{\top}, \,\allowbreak\ldots, \,\allowbreak(u^{k+T-1})^{\top}]^{\top}$, is calculated as $E_{33}$. Thus, by choosing the matrix $R$ as positive definite, the Lagrangian $\mathcal{L}$ becomes convex, reaching its global minimum using the proposed optimal control input $u^{*|k}$. 
\end{remark}

\begin{remark}{\bf (Existence of Optimal Control Input)}
The optimal control input \eqref{eqn: optimal_u}  \textbf{always} exists regardless of the real-valued system matrices $A$, $B$, and $C$ in \eqref{eqn: LTI system}. This is because the invertibility of $E$ depends only on $E_{11}$ and $E_{33}$, which are independent of the system matrices, whereas $E_{12}$ is known to be invertible regardless of $A$ due to its structure. This guarantees that the optimal control input \eqref{eqn: optimal_u} can be found for any LTI system with real matrices.
\end{remark}

\subsubsection{Nonlinear system case}
Consider a discrete-time control-affine nonlinear system as follows:
\begin{align}
    \mathsf{x}^{k+1} = \mathsf{f}(\mathsf{x}^k)+\mathsf{g}(\mathsf{x}^k)u^k,\label{eqn: affine nonlinear system}
\end{align}
where $\mathsf{f}:\mathbb{R}^{n} \rightarrow \mathbb{R}^{n}$ and $\mathsf{g}:\mathbb{R}^{n} \rightarrow \mathbb{R}^{n\times m}$ are nonlinear transformations. The Lagrangian associated with the local Wasserstein distance and the nonlinear dynamics is formulated by
\begin{align}
    \begin{aligned}
         \mathcal{L} = &\Phi(\mathsf{x}^{k+T})+\sum_{i=k}^{k+T-1}\left(\dfrac{1}{2}(\mathcal{W}^{i|k})^2 + \dfrac{1}{2}(u^{i})^{\top}Ru^{i}+\right.\\
         &\left.\dfrac{1}{2}(\mathsf{x}^{i})^{\top}Q\mathsf{x}^{i}+(\lambda^{i+1})^{\top}(\mathsf{f}(\mathsf{x}^{i}) + \mathsf{g}(\mathsf{x}^{i})u^{i}-\mathsf{x}^{i+1})\right),
         \end{aligned}\label{eqn: Lagrangian_NonLin}
    \end{align}
where $Q\succeq \mathbf{0}$ and $R\succ \mathbf{0}$. The necessary conditions and boundary conditions to minimize \eqref{eqn: Lagrangian_NonLin} are obtained by
{\allowdisplaybreaks
\begin{subequations}
\begin{align}
    &\bar{Q}\mathsf{x}^{*|i}-\lambda^{*|i}+ \left[\dfrac{\partial\mathsf{f}^{i}}{\partial\mathsf{x}^{i}}\right]^{\top}\lambda^{*|i+1} = (\sum\nolimits_{j \in \mathcal{S}^{k}} \gamma_j)C^{\top}\overline{q}^{k}, \label{eqn: NLCAS_pontryagin-a}
    \\&\mathsf{f}^{*|i'} = \mathsf{x}^{*|i'+1}, \label{eqn: NLCAS_pontryagin-b}
    \\&Ru^{*|i'} + \mathsf{g}\label{eqn: NLCAS_pontryagin-c}(\mathsf{x}^{*|i'})^{\top}\lambda^{*|i'+1}=\mathbf{0},
    \\&\qquad \left(
        \begin{aligned}
        &\text{for } i =  k+1,  k+2, \cdots ,k+T-1
        \\&\text{and } i'=k,k+1,\cdots,k+T-1 
        \end{aligned} \right)\nonumber
    \\&\bar{Q}\mathsf{x}^{*|k+T}- \lambda^{*|k+T} = (\sum\nolimits_{j \in \mathcal{S}^{k}} \gamma_j)C^{\top}\overline{q}^{k},\label{eqn: NLCAS_pontryagin-d}
    \\&\frac{\partial^2 \mathcal{L}}{(\partial \mathbf{u}^k)^2}(\bar{\x}^{k+1},\bar{\lambda}^{k+1},\bar{u}^{k})\succ \mathbf{0},\label{eqn: NLCAS_pontryagin-e}
    \end{align}
\end{subequations}}where $\mathsf{f}^{i}=\mathsf{f}(\mathsf{x}^{i}) + \mathsf{g}(\mathsf{x}^{i})u^{i}$, $\bar{Q}=\{(\sum_{j \in \mathcal{S}^{k}} \gamma_j)C^{\top}C+ Q\}$, $\overline{q}^{k}=(\sum_{j \in \mathcal{S}^{k}} \gamma_j q_{j})/(\sum_{j \in \mathcal{S}^{k}} \gamma_j)$, and $\mathbf{u}^k = \allowbreak[(u^{k})^{\top},\,\allowbreak\ldots, \,\allowbreak(u^{k+T-1})^{\top}]^{\top}$. The equation \eqref{eqn: NLCAS_pontryagin-e} ensures that the control input $u$ corresponds to a local minimizer (given $\lambda$), while the costate $\lambda$ satisfies the backward adjoint equations \eqref{eqn: NLCAS_pontryagin-a}, which correspond to maximizing the Lagrangian with respect to $\lambda$. Since the left-hand side of \eqref{eqn: NLCAS_pontryagin-e} simplifies to $R \otimes \mathbf{I}_{T}$, the condition is satisfied by choosing $R \succ \mathbf{0}$.

Unlike the LTI system case, the solution for the optimal control input cannot be obtained analytically due to the nonlinearity. Thus, a numerical solver is necessary to obtain the optimal control input for the general horizon length $T$ greater than one. For the horizon length $T=1$, however, we have at least the following result.

\begin{proposition}\label{prop: nonlinear optimal control}
Consider the control-affine nonlinear system represented by \eqref{eqn: affine nonlinear system}. The optimal control input that forms a primal-dual solution to the Lagrangian in \eqref{eqn: Lagrangian_NonLin} over a single horizon is derived by
\begin{align}
    &u^{*|k} = K\{(\sum\nolimits_{j \in \mathcal{S}^{k}} \gamma_j)C^{\top}\overline{q}^{k}-\bar{Q}\mathsf{f}(\mathsf{x}^{k}) \},\label{eqn: Optimal input of affine nonlinear system}
    \\&K = \left \{R+\mathsf{g}(\mathsf{x}^k)^{\top}\bar{Q}\mathsf{g}(\mathsf{x}^k)\right\}^{-1}\mathsf{g}(\mathsf{x}^k)^\top.\nonumber
\end{align}
\end{proposition}
\begin{proof}
   For the control-affine nonlinear system, the necessary conditions for optimality over a single horizon can be derived from the saddle point condition of the Lagrangian. That is, the control input \( u^{*|k} \) minimizes the Lagrangian, while the costate \( \lambda^{*|k+1} \) maximizes it. The corresponding stationarity conditions are given by:
\begin{subequations}
    \begin{align}
&\mathsf{x}^{*|k+1} = \mathsf{f}(\mathsf{x}^k)+\mathsf{g}(\mathsf{x}^k)u^{*|k},\label{eqn: Affine_SingleHor_a}
\\&Ru^{*|k} + \mathsf{g}(\mathsf{x}^k)^{\top}\lambda^{*|k+1}=\mathbf{0},\label{eqn: Affine_SingleHor_b}
\\&\bar{Q}\mathsf{x}^{*|k+1}- \lambda^{*|k+1} = (\sum\nolimits_{j \in \mathcal{S}^{k}} \gamma_j)C^{\top}\overline{q}^{k}.\label{eqn: Affine_SingleHor_c}
\end{align}
\end{subequations}
These conditions characterize the saddle point of the Lagrangian \( \mathcal{L}(\x^{*|k+1}, \lambda^{*|k+1}, u^{*|k}) \), where \( u^{*|k} \) solves \( \min_u \mathcal{L} \), and \( \lambda^{*|k+1} \) solves \( \max_\lambda \mathcal{L} \). 

To guarantee that the stationary point corresponds to a local minimum with respect to the control input, the second variation of the Lagrangian with respect to \( u^k \) must be positive definite, i.e.,
\begin{align}
\frac{\partial^2 \mathcal{L}}{(\partial u^k)^2}(\x^{*|k+1},\lambda^{*|k+1},u^{*|k})=R\succ \mathbf{0}.\label{eqn: NLCAS minimum guarantee}
\end{align}
Since \( R \) is assumed to be a positive definite matrix in \eqref{eqn: Lagrangian_NonLin}, the condition \eqref{eqn: NLCAS minimum guarantee} ensures that \( u^{*|k} \) is a local minimizer of the Lagrangian in the \( u^k \)-direction.

Then, by solving equations \eqref{eqn: Affine_SingleHor_a}-\eqref{eqn: Affine_SingleHor_c} through substitution, the optimal control input for D$^2$OC of the nonlinear system is obtained by \eqref{eqn: Optimal input of affine nonlinear system}.
\end{proof}

\begin{remark} \textbf{(Local vs. Global Optimality in Nonlinear Control)}
In Proposition \ref{prop: nonlinear optimal control}, the optimal control law was derived using the method of Lagrange multipliers, considering the control-affine nonlinear dynamics of the system given by \eqref{eqn: affine nonlinear system}.
The resulting control law satisfies the necessary conditions for optimality, ensuring \textit{local optimality} of the solution. It is important to note that, due to the nonlinearity of the system, the control law may not always guarantee \textit{global optimality}. Multiple local optima can exist for such nonlinear systems, which limits the scope of the derived controller's optimality. The proposed controller ensures that it satisfies the necessary conditions for a local optimal solution, but does not rule out the existence of other potential local optima.
\end{remark}

\begin{remark} \textbf{(Semi-Analytic Solution for Nonlinear Control)}
In Proposition 2, while the computation of the optimal control input involves the nonlinear terms $\mathsf{f}(\x^k)$ and $\mathsf{g}(\x^k)$, it is important to note that Equation \eqref{eqn: Optimal input of affine nonlinear system} provides a semi-analytic optimal solution. The only numerical computation required is the inversion of a matrix. Since the matrix $R$ is positive-definite, the term inside the inverse is well-posed and invertible, ensuring that the optimal gain $K$ can be computed efficiently for any given nonlinear term $\mathsf{g}(\x^k)$. Thus, the implementation of the control input is practical and straightforward even in the presence of nonlinearities.
\end{remark}

\subsection{Weight Update Stage}
The optimal control input is applied to agent $r$ at time $k$, and the agent then moves to the new coordinate ${}^{r}y^{k+1}$ at time $k+1$. In this stage, the weight of the sample-points should be updated accordingly by transporting their weights to the newly created agent-point at ${}^{r}y^{k+1}$. Through this stage, the sample-points will lose their weight, diminishing their priority at a later time. Consequently, the agent will more likely visit other sample-points that are unexplored according to the proposed local sample-point selection formula \eqref{eqn: weight-normalized distance} as they have higher weights. 

This weight update stage determines how much weight should be transported from each sample-points to the agent-point. 
The solution to this problem can be obtained by solving the following optimal transport problem:
\begin{align}\label{eqn: Weight update problem}
&{}^{r}\gamma_{j}^{*|k+1} = \argmin_{{}^{r}\gamma_{j}^{k+1}} \sum\nolimits_{j} {}^{r}\gamma_{j}^{k+1}\lVert {}^{r}y^{k+1} - q_j \rVert^2 \\
		&\text{subject to} \quad
  \begin{aligned}
  & \sum\nolimits_{j} {{}^{r}\gamma}^{k+1}_{j} = {}^r\alpha^{k+1},\  0 \leq {{}^{r}\gamma}^{k+1}_{j} \leq  {{}^{r}\beta}^k_{j},\quad \forall j,
    \end{aligned}\nonumber
	\end{align}
where ${}^{r}\gamma_{j}^{*|k+1}$ is the optimal transportation from the sample-point to the agent-point ${}^{r}y^{k+1}$. The first constraint implies that the total transported weight must be the same as the demand capacity on the agent-point, which is calculated by the consumed energy between the time interval $[k, k+1]$ divided by the initially available energy.
The second one means the transportation must be nonnegative as well as upper-bounded by the supply capacity ${}^{r}\beta^k_{j}$ from the corresponding sample-point.

This optimal transportation problem \eqref{eqn: Weight update problem} is a Linear Programming (LP) problem for which the solution can be obtained by using an LP solver. Then, the remaining weight of the sample-points is updated as
\begin{align}
	{}^{r}\beta^{k+1}_{j} &= {}^{r}\beta^{k}_{j}- {}^{r}{\gamma}_{j}^{*|k+1},\quad \forall j.\label{eqn: weight update}
\end{align}

\subsection{Weight-Sharing Stage for Decentralized Control}
Stages A and B are applied independently to each agent in the multi-agent system. However, collaborations between the agents have not been considered thus far. To accomplish multi-agent collaborative coverage in a decentralized setup, Stage C performs the information exchange necessary to synchronize the agents.

Suppose that any two agents are within the communication range threshold $r_\text{comm}$ and hence they can share the weight information through communications. In that case, the agents (e.g., agents $r$ and $s$) can exchange the coverage progress so that each of them can know which areas in the domain were already covered by other agents.

For example, in the previous work \cite{lee2022density}, the following weight-sharing method was proposed.

\noindent\textit{Original weight-sharing method}:
\begin{align}\label{eqn: D2C weight update}
    {{}^r\beta}^{k}_{j} = {^s\beta}^{k}_{j} = \min({^r\beta}^{k}_{j},\,{^s\beta}^{k}_{j})\,\, \forall j, \,\, \text{if } \|{}^ry^k - {}^sy^k\| \leq r_{\text{comm}},
\end{align}
which enables the agents to update the weight at each sample-point by adopting the minimum value between agents $r$ and $s$, provided that the distance between them is within the communication range $r_{\text{comm}}$ at time $k$.

It turns out that this weight-sharing method may lead to some work redundancy, causing lower efficiency in area coverage with decentralized control. To prevent this work redundancy, a new weight-sharing method is proposed as follows.

\noindent\textit{Proposed weight-sharing method}:
\begin{align}
	&{}^{r|l}{\Gamma}_{j}^{n+1}={}^{s|l}{\Gamma}_{j}^{m+1}=\max({}^{r|l}{\Gamma}_{j}^{n},{}^{s|l}{\Gamma}_{j}^{m}),\,\, \forall j,l=1,...,L ,\nonumber
\\&{}^r\beta^{k}_j={\beta}_{j}^{0}  - \sum\nolimits_{l=1}^{L} {}^{r|l}{\Gamma}_{j}^{n+1}, \,\, \text{if } \|{}^ry^k - {}^sy^k\| \leq r_{\text{comm}},\label{eqn: D2OC weight update}
\end{align}
where $l$ denotes the agent index, and $\beta_j^{0}$ is the initial weight of the sample-point $q_j$.
Although the proposed D$^2$OC method can handle non-uniform initial weights for the sample-points, a uniform initial weight is considered in this paper for simplicity. Instead, the priority of the domain in a given mission is reflected by sampling more points in high-priority areas.

In \eqref{eqn: D2OC weight update}, ${}^{r|l}\Gamma^n_j$ represents the \textit{coverage progress} defined by
\begin{align}
\ {}^{r|r}\Gamma^n_j=\sum_{i=0}^{k} {}^{r}{\gamma}_{j}^{*|i},\quad {}^{r|l}\Gamma^n_j=\sum_{i=0}^{k^{l\rightarrow r}} {}^{l}{\gamma}_{j}^{*|i}, \,\, \forall l \neq r, \label{eqn: coverage progress}
\end{align}
where the superscript $n\in \mathbb{N}_0$ indicates agent $r$ has completed the weight-sharing $n$ times, and the variable $k^{l\rightarrow r}$ represents the time up to which agent $r$ is aware of the coverage progress achieved by agent $l$.
In short, ${}^{r|l}{\Gamma}_{j}^{n}$ represents the coverage progress on the sample-point $q_j$ done by agent $l$, which is transferred to agent $r$, not only directly by agent $l$ but also indirectly by other agents.

Fig. \ref{fig: schematic_proposed_sharing} exemplifies the information exchange process for the proposed weight-sharing method. Consider that three agents are employed, and no weight sharing has occurred yet. At time $k_1$, the first weight sharing occurred between agents 2 and 3 as they were within the communication range threshold, and they exchanged their own coverage progress. At time $k_2$, the second weight sharing occurred between agents 1 and 3. At this point, they exchange not only their own coverage progress but also coverage progress achieved by agent $2$.
Notice that at time $k_2$ the coverage progress ${}^{3|3}{\Gamma}_{j}^{1}$ is given by 0.2 while it was ${}^{3|3}{\Gamma}_{j}^{0}=0.1$ at time $k_1$, which is not because of the communication but due to its own work.

\begin{figure}[h]
    \centering
    \includegraphics[width=0.95\linewidth]{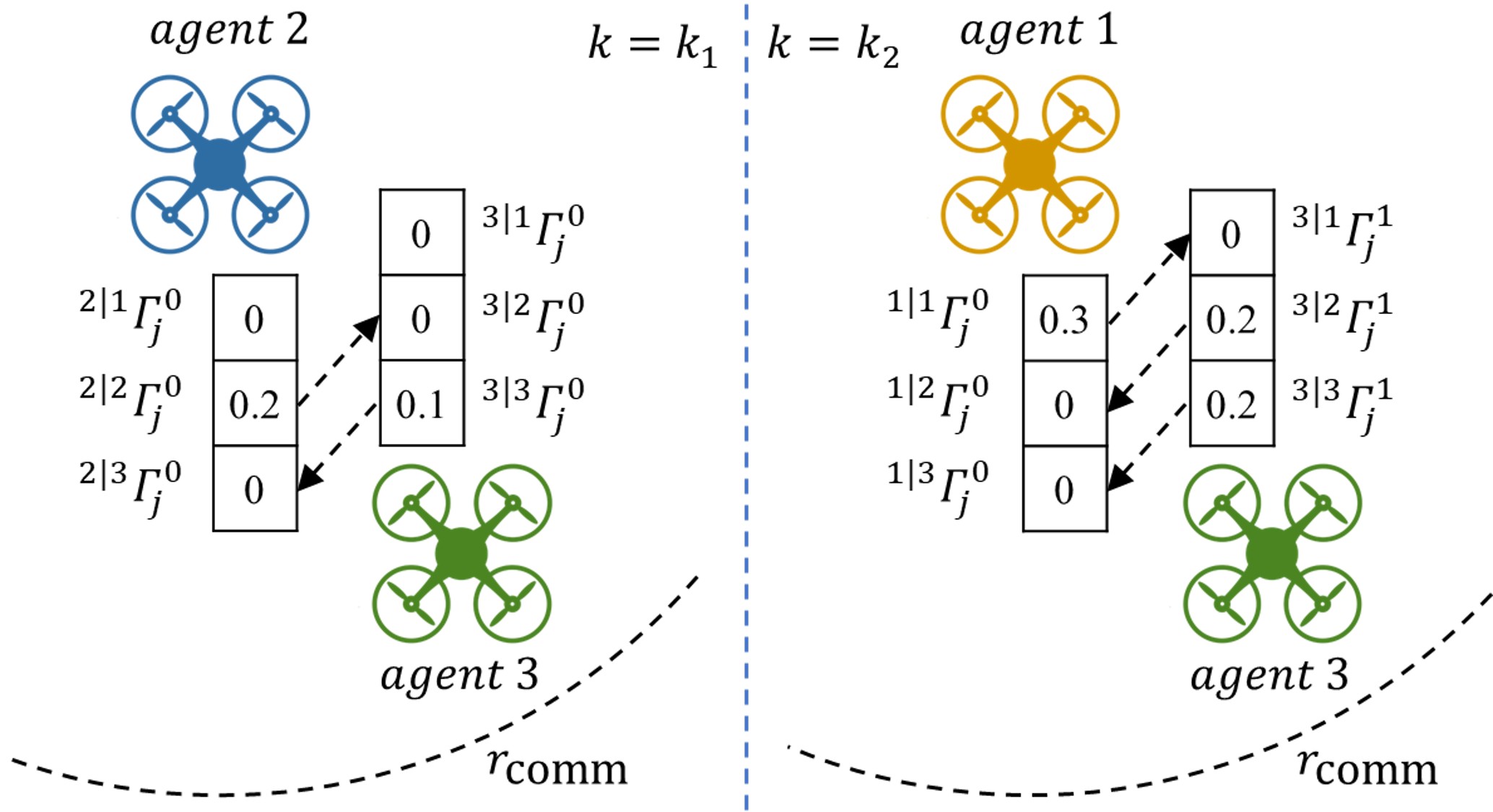}
    \caption{Schematic illustrating the exchange of the coverage progress in the proposed weight-sharing method.
    }
    \label{fig: schematic_proposed_sharing}
\end{figure}

The coverage progress for all agents is updated based on the most recent information, following the first equation in \eqref{eqn: D2OC weight update}.
The second equation denotes that the remaining weight of the sample-point is adjusted using the updated coverage progress during Stage C.

\begin{figure}[h]
    \centering
    \includegraphics[width=0.95\linewidth]{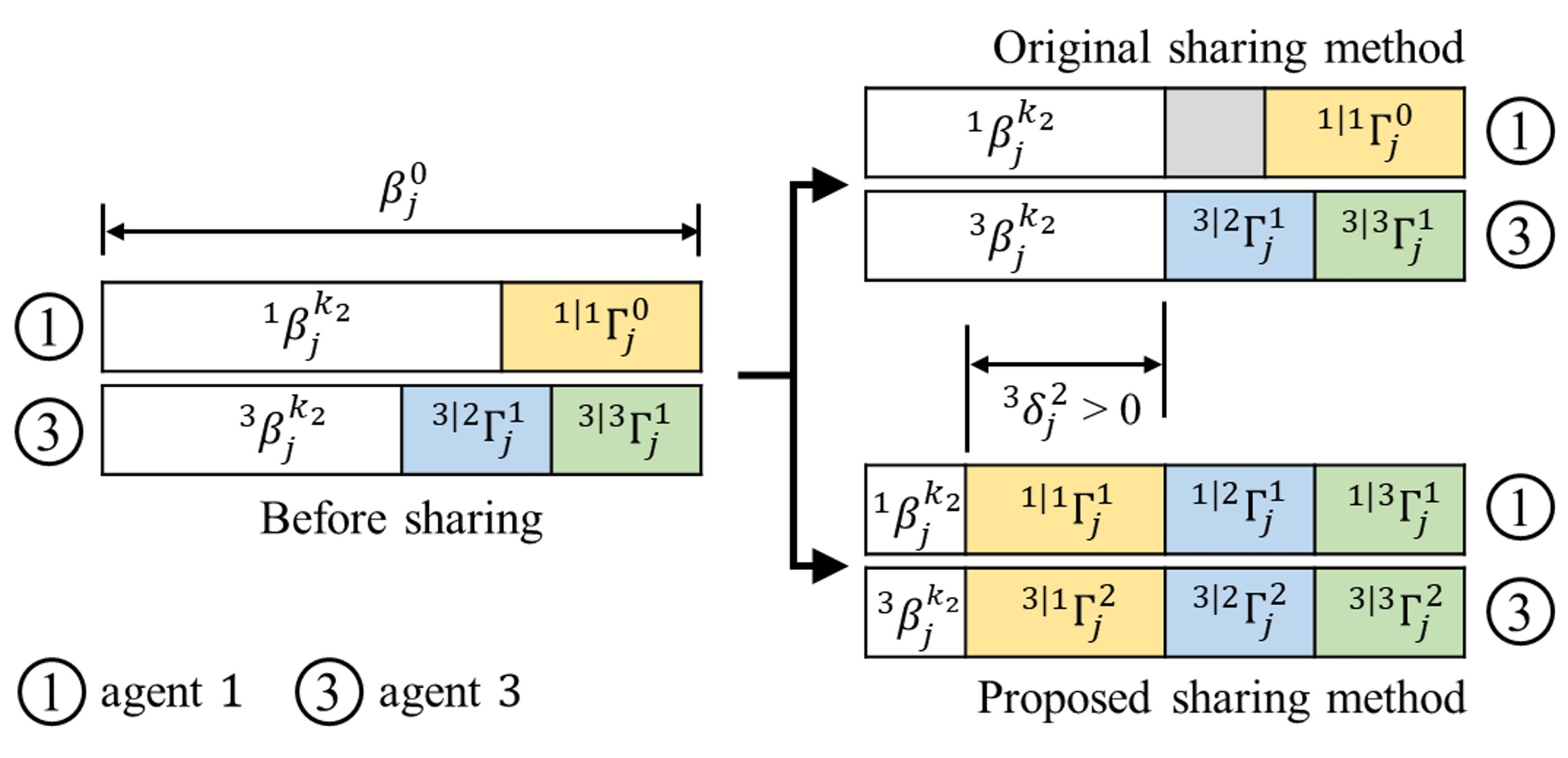}
    \caption{Schematic illustrating both weight-sharing methods for the scenario of $k=k_2$ in Fig. \ref{fig: schematic_proposed_sharing}.} 
    \label{fig: weight sharing comparison}
\end{figure}
To compare the performance between the original and proposed weight-sharing methods, consider the case where the original sharing method was applied. Similar to the expression for the remaining weight of the proposed weight-sharing method in \eqref{eqn: D2OC weight update}, the remaining weight of the original method can be alternatively represented by
   \begin{equation}
    {}^r\beta^{k}_j={\beta}_{j}^{0}  - \sum\nolimits_{l=1}^{L} {}^{r|l}{\Gamma}_{j}^{n} + {}^r\delta^{n}_{j},
    \label{eqn: coverage progress - original}
    \end{equation}
where the symbol ${}^r\delta^{n}_{j}\in\mathbb{R}$ is adopted to represent the difference in the remaining weight on $q_j$ compared to the proposed sharing method using \eqref{eqn: D2OC weight update}. The sign of ${}^r\delta^{n}_{j}$ then determines which weight-sharing method has less (or more) remaining weight on the same sample-point.
For instance, ${}^r\delta^{n}_{j} > 0$ indicates that the original weight-sharing method results in a higher remaining weight and hence, lower efficiency.

Fig. \ref{fig: weight sharing comparison} depicts the remaining weights after applying each method at the moment of $k=k_2$ in Fig. \ref{fig: schematic_proposed_sharing}. In this scenario, the sign of ${}^{3}\delta^2_j$ is greater than 0, showing the proposed weight-sharing method is more efficient than the original one.

To rigorously prove the sign of ${}^r\delta^{n}_{j}$ for a more general case, the following results are provided.

\begin{lemma}\label{lemma: minimum equation}
For any $A,B \in \mathbb{R}$, the minimum of the two, i.e., $\min(A,B)$, is obtained by
\begin{equation}
    \min (A,B) = \left(A+B-|A-B|\right)/2.\nonumber
\end{equation}
\end{lemma}
\begin{proposition}\label{prop: general form minimum share}
Consider the decentralized density-driven optimal control strategy applied to agent $r$ in the multi-agent system with the original weight-sharing method \eqref{eqn: D2C weight update}.
For agent $r$, the current remaining weight on the sample-point $q_j$ is given by \eqref{eqn: coverage progress - original}, with an additive term ${}^{r}\delta_j^n\in\mathbb{R}$.

If agent $r$ can communicate with agent $s$ within the given communication range threshold, $r_\text{comm}$, the remaining weight on the sample-point $q_j$ for agent $r$ is updated by 
    \begin{equation}
    {}^r\beta^{k}_j={\beta}_{j}^{0}  - \sum\nolimits_{l=1}^{L} {}^{r|l}{\Gamma}_{j}^{n+1} + {}^r\delta^{n+1}_{j},
    \label{eqn: general form D2C}
    \end{equation}
where the weight-sharing count is increased by one ($n\rightarrow n+1$) as agent $s$ has shared the weight information with agent $r$.

Moreover, ${}^r\delta^n_j$, referred to as \textit{the omission of the coverage progress} hereafter, is always nonnegative for any $n\in\mathbb{N}_0$.

\end{proposition}
\begin{proof}
For agent $r$, the remaining weight when no weight sharing has occurred ($n=0$) is computed using Equation \eqref{eqn: weight update} by 
\begin{align}
{}^r\beta^{k}_j={\beta}_{j}^{0} -\sum\nolimits_{i=0}^{k} {}^{r}{\gamma}_{j}^{*|i} ={\beta}_{j}^{0} - {}^{r|r}{\Gamma}^{0}_{j}.\label{eqn: unshared_remaining_weight}
\end{align}
By directly comparing \eqref{eqn: coverage progress - original} and \eqref{eqn: unshared_remaining_weight}, it can be shown that ${}^{r}{\delta}_j^{0}=0$.

For $n>0$, consider agent $r$ (agent $s$) has engaged in the weight sharing $n$ times ($m$ times). The remaining weight for agents $r$ and $s$, respectively, are represented by
\begin{align*}
    {}^r\beta^{k}_j={\beta}_{j}^{0}  - \sum_{l=1}^{L} {}^{r|l}{\Gamma}_{j}^{n} + {}^r\delta^{n}_{j}
    ,\ {}^s\beta^{k}_j={\beta}_{j}^{0}  - \sum_{l=1}^{L} {}^{s|l}{\Gamma}_{j}^{m} +{}^s\delta^{m}_{j}.
\end{align*}
If these two agents update the weight information on $q_j$ using the original weight-sharing method in \eqref{eqn: D2C weight update}, the remaining weight on $q_j$ for agent $r$ is computed by
{\allowdisplaybreaks
\begin{align*}
    &{}^r\beta^{k}_j=\min\left({}^r\beta^{k}_j, {}^s\beta^{k}_j\right) = \frac{1}{2}\left({}^r\beta^{k}_j+{}^s\beta^{k}_j-|{}^r\beta^{k}_j-{}^s\beta^{k}_j|\right)
    \\&\begin{aligned}\overset{(a)}{=}&{\beta}_{j}^{0} - \sum_{l=1}^{L} \left\{{}^{r|l}\Gamma^{n}_j+{}^{s|l}\Gamma^{m}_j\right\}+\frac{1}{2}\left\{\sum_{l=1}^{L} \left\{{}^{r|l}\Gamma^{n}_j+{}^{s|l}\Gamma^{m}_j\right\}\right.
    \\&+{}^{r}\delta^{n}_{j}+{}^{s}\delta^{m}_{j}-\left.\left|- \sum_{l=1}^{L} {}^{r|l}{\Gamma}_{j}^{n} + {}^r\delta^{n}_{j}+ \sum_{l=1}^{L} {}^{s|l}{\Gamma}_{j}^{m} - {}^s\delta^{m}_{j}\right|\right\}\end{aligned}
    \\&\overset{(b)}{=}{\beta}_{j}^{0} - \sum_{l=1}^{L} \max\left({}^{r|l}{\Gamma}_{j}^{n},{}^{s|l}{\Gamma}_{j}^{m}\right)- \sum_{l=1}^{L} \min\left({}^{r|l}{\Gamma}_{j}^{n},{}^{s|l}{\Gamma}_{j}^{m}\right)
    \\&\quad+\frac{1}{2}\left\{\sum_{l=1}^{L} {}^{s|l}{\Gamma}_{j}^{m}+{}^{r}\delta^{n}_{j}+\sum_{l=1}^{L} {}^{r|l}{\Gamma}_{j}^{n}+{}^{s}\delta^{m}_{j}\right.
    \\&\quad-\left.\left|\sum_{l=1}^{L} {}^{s|l}{\Gamma}_{j}^{m} + {}^r\delta^{n}_{j}- \left(\sum_{l=1}^{L} {}^{r|l}{\Gamma}_{j}^{n}  + {}^s\delta^{m}_{j}\right)\right|\right\}
    \\
    &\overset{(c)}{=}{\beta}_{j}^{0} - \sum_{l=1}^{L} {}^{r|l}{\Gamma}_{j}^{n+1}- \sum_{l=1}^{L} \min\left({}^{r|l}{\Gamma}_{j}^{n},{}^{s|l}{\Gamma}_{j}^{m}\right)
    \\&\quad+\min\left(\sum_{l=1}^{L} {}^{s|l}{\Gamma}_{j}^{m}+{}^r\delta^{n}_{j},\sum_{l=1}^{L} {}^{r|l}{\Gamma}_{j}^{n} +  {}^s\delta^{m}_{j}\right)
    \\&={\beta}_{j}^{0} - \sum_{l=1}^{L} {}^{r|l}
    {\Gamma}_{j}^{n+1}+{}^r\delta^{n+1}_{j}.
    \nonumber
\end{align*}
}
In deriving the equality (a), the equation \eqref{eqn: coverage progress - original} is applied. 
For the equality (b), ${}^{r|l}\Gamma^{n}_j+{}^{s|l}\Gamma^{m}_j$ was rewritten by $\max\left({}^{r|l}{\Gamma}_{j}^{n},{}^{s|l}{\Gamma}_{j}^{m}\right)+\min\left({}^{r|l}{\Gamma}_{j}^{n},{}^{s|l}{\Gamma}_{j}^{m}\right)$. Finally, in deriving the equality (c), 
both the first equation in \eqref{eqn: D2OC weight update} and Lemma \ref{lemma: minimum equation} were utilized. 

Then, the omission of the coverage progress at the sharing count $n+1$, denoted by ${}^r\delta^{n+1}_{j}$, leads to
{\allowdisplaybreaks
\begin{align}
& \begin{aligned}{}^r\delta^{n+1}_{j}=&- \sum\nolimits_{l=1}^{L} \min\left({}^{r|l}{\Gamma}_{j}^{n},{}^{s|l}{\Gamma}_{j}^{m}\right)
    \\&+\min\left(\sum\nolimits_{l=1}^{L} {}^{s|l}{\Gamma}_{j}^{m}+{}^r\delta^{n}_{j},\sum\nolimits_{l=1}^{L} {}^{r|l}{\Gamma}_{j}^{n} +  {}^s\delta^{m}_{j}\right)\end{aligned}\nonumber
    \\&=\min\left(\sum\nolimits_{l=1}^{L}\left\{{}^{s|l}{\Gamma}_{j}^{m}-\min\left({}^{r|l}{\Gamma}_{j}^{n},{}^{s|l}{\Gamma}_{j}^{m}\right)\right\}+{}^r\delta^{n}_{j}\right.,\nonumber
    \\&\quad\left.\sum\nolimits_{l=1}^{L} \left\{{}^{r|l}{\Gamma}_{j}^{n}-\min\left({}^{r|l}{\Gamma}_{j}^{n},{}^{s|l}{\Gamma}_{j}^{m}\right)\right\}+{}^s\delta^{m}_{j}\right)\nonumber
    \\&\geq \min\left({}^r\delta^{n}_{j},{}^s\delta^{m}_{j}\right).
    \nonumber    
\end{align}
}
Since it is already proved ${}^r\delta^{0}_{j}=0$ (similarly ${}^s\delta^{0}_{j}=0$) at the beginning of this proof process, it is guaranteed based on the above inequality that ${}^r\delta^{n}_{j} \geq 0$ for any $r$ and $n\in\mathbb{N}_0$.
\end{proof}
Following Proposition \ref{prop: general form minimum share}, the inequality relationship of the remaining weight for the centralized case, the original sharing method \eqref{eqn: D2C weight update}, and the proposed sharing method \eqref{eqn: D2OC weight update} is established as follows.
\begin{theorem}\label{theorem: proposed sharing method} 

For the density-driven optimal control strategy, consider the three distinct weight-sharing methods: centralized, original \eqref{eqn: D2C weight update}, and proposed \eqref{eqn: D2OC weight update}, where the remaining weight of the sample-point $q_j$ at time $k$ for the centralized case is represented by
\begin{align}
\beta^{k}_j|_{\text{centralized}}={\beta}_{j}^{0}  - \sum\nolimits_{l=1}^{L} {}^{l|l}{\Gamma}_{j}^{n}. \label{eqn: actual remaining weight}
\end{align}

Then, the remaining weight when applying each method has the following relationship.
\begin{align}
\beta^{k}_j|_{\text{centralized}}\leq{}^{r}\beta^{k}_j|_{\text{proposed}}\leq{}^{r}\beta^{k}_j|_{\text{original}} \label{eqn: comp_sharing}
\end{align}
\begin{proof}
The second inequality in \eqref{eqn: comp_sharing} can be shown according to \eqref{eqn: D2OC weight update} and Proposition \ref{prop: general form minimum share} by
\begin{align}
\underbrace{{\beta}_{j}^{0}  - \sum\nolimits_{l=1}^{L} {}^{r|l}{\Gamma}_{j}^{n}}_{{}^{r}\beta^{k}_j|_{\text{proposed}}}\leq\underbrace{{\beta}_{j}^{0}  - \sum\nolimits_{l=1}^{L} {}^{r|l}{\Gamma}_{j}^{n}+{}^{r}{\delta}^{n}_{j}}_{{}^{r}\beta^{k}_j|_{\text{original}}}, \nonumber
\end{align}
since ${}^{r}{\delta}^{n}_{j}\geq 0$, $\forall n\in\mathbb{N}_0$.

For the first inequality, it is obvious that, in the centralized case, the summation of the coverage progress, $\sum_{l=1}^{L} {}^{l|l}{\Gamma}_{j}^{n}$, is greater than or equal to $\sum_{l=1}^{L} {}^{r|l}{\Gamma}_{j}^{n}$ for the proposed decentralized weight-sharing method. This is because, in the centralized one, each agent has access to the weight information on $q_j$ from all other agents at each discrete time.

As a consequence, the following relationship holds.
\begin{align}
\underbrace{{\beta}_{j}^{0}  - \sum_{l=1}^{L} {}^{l|l}{\Gamma}_{j}^{n}}_{\beta^{k}_j|_{\text{centralized}}}\leq\underbrace{{\beta}_{j}^{0}  - \sum_{l=1}^{L} {}^{r|l}{\Gamma}_{j}^{n}}_{{}^{r}\beta^{k}_j|_{\text{proposed}}}\leq\underbrace{{\beta}_{j}^{0}  - \sum_{l=1}^{L} {}^{r|l}{\Gamma}_{j}^{n}+{}^{r}{\delta}^{n}_{j}}_{{}^{r}\beta^{k}_j|_{\text{original}}} \nonumber
\end{align}
\end{proof}
\end{theorem}

Based on Theorem \ref{theorem: proposed sharing method}, it is guaranteed that the proposed weight-sharing method \eqref{eqn: D2OC weight update} will yield a lower remaining weight than the original method, achieving less work redundancy and hence, higher coverage efficiency.

\begin{remark}{\bf (Asynchronous Communication)}
    The proposed D$^2$OC is executable asynchronously, making it highly practical for real-world multi-agent systems. Unlike many existing multi-agent systems that require synchronized communications and a global clock, D$^2$OC allows agents to function independently without the need for constant communication. Communication between agents only occurs infrequently when they are within each other's communication range. Moreover, the amount of data exchanged is minimal, consisting of small updates on sample-point weights, ensuring that even if communication occurs, it does not significantly affect system performance. This reduces concerns about communication delays or packet drops, making the system robust to asynchrony.
\end{remark}

\section{Simulation}\label{sec: sim}
To evaluate the performance of the D$^2$OC scheme as well as the proposed weight-sharing method, different types of simulations were carried out. The trajectories of the multi-agent system are presented by applying D$^2$OC, SMC, and D$^2$C schemes for the qualitative comparison, while the 2-Wasserstein distance was calculated as a metric for the quantitative comparison.
The energy flexibility of the proposed D$^2$OC scheme is also shown as the strength of D$^2$OC.
Moreover, the remaining weight comparison between the original and proposed weight-sharing methods is presented to validate Theorem \ref{theorem: proposed sharing method}.
\subsection{Performance Comparison between SMC and D$^2$OC}
In this subsection, two separate simulations were conducted: one for a simple symmetric reference distribution to validate and discuss the energy flexibility of the proposed control scheme and another for a more complicated reference distribution to evaluate the performance of the D$^2$OC scheme compared to the SMC method.

Although the proposed D$^2$OC scheme can be applied to both linear and nonlinear systems, the following first-order integrator $\mathsf{x}^{k+1} = {\mathsf{x}}^{k} + {u}^{k}$ is used for the performance comparison between SMC and D$^2$OC.

The major reason to adopt the simple first-order integrator is that the controller in SMC was developed mainly for the first- and second-order systems. A more complicated but realistic model (linearized quadrotor model) is employed for comparison with another controller in the next subsection.

In Fig. \ref{fig: SimA_energy flexibility}, the contour plot of the reference density function is presented, and the sample-points were generated accordingly.
A two-agent first-order integrator system was considered, where their initial positions are marked with the `x' symbol in Fig. \ref{fig: SimA_energy flexibility}.
Some parameters used in this simulation are shown in Table \ref{table: A parameter}.
Figs. \ref{fig: SimA_energy flexibility}(a)-(b) present the trajectories of two agents with D$^2$OC, and Figs. \ref{fig: SimA_energy flexibility}(c)-(d) are for the SMC method with 30 Fourier basis functions (notice that in the SMC method \cite{mathew2011metrics}, the infinite Fourier basis functions are required; however, the truncation is inevitable for the implementation). The operation time (or terminal time) of the two-agent system is set differently: 20 seconds for (a) and (c), and 200 seconds for (b) and (d).

\begin{figure}[h]
    \centering
    \includegraphics[width=0.90\linewidth]{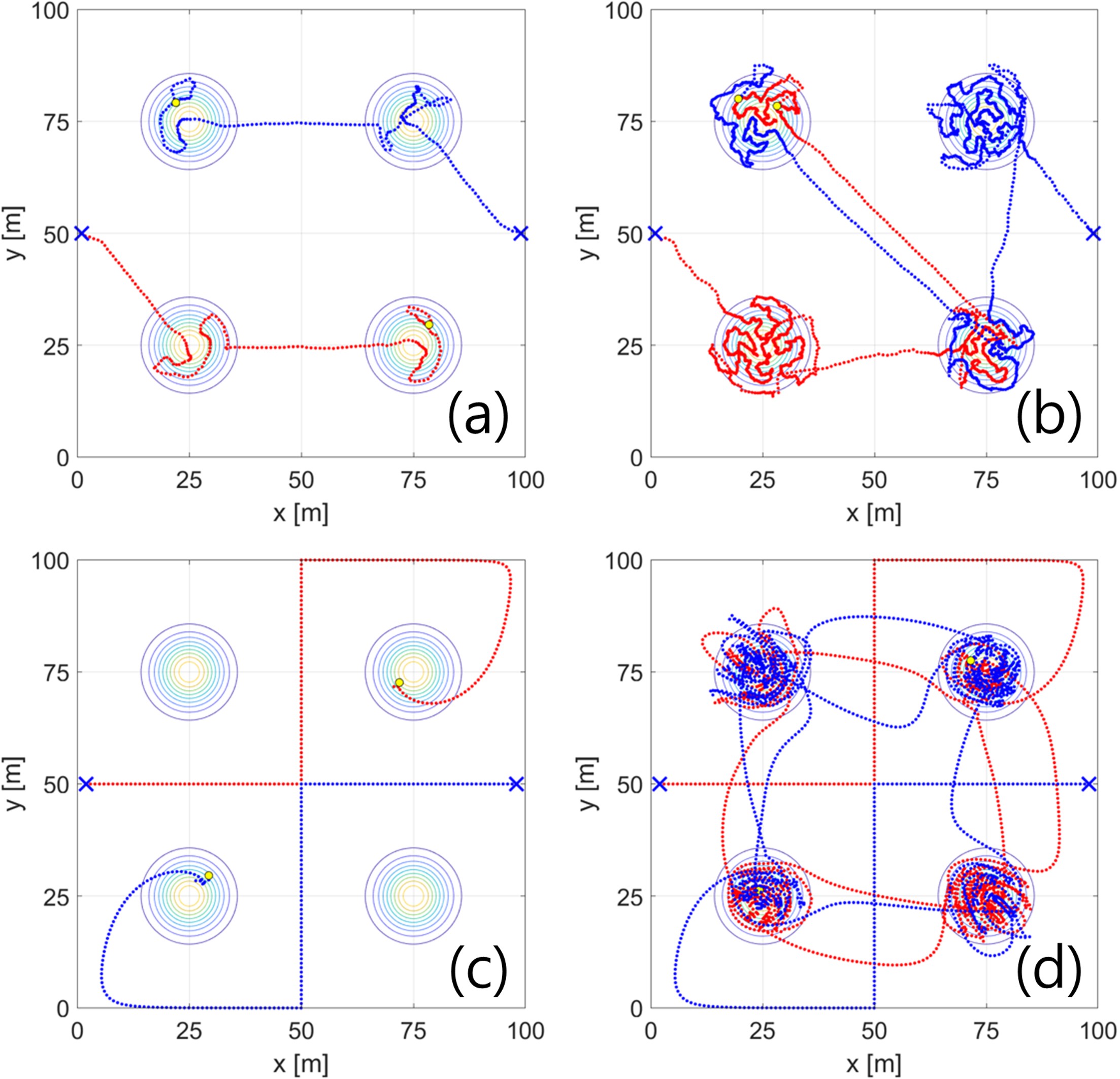}
    \caption{Trajectory comparison between the SMC method and the D$^2$OC scheme. The initial and final locations of agents are indicated by blue-cross and yellow-circle marks, respectively: (a) D$^2$OC with a terminal time of 20 seconds; (b) D$^2$OC with a terminal time of 200 seconds; (c) SMC with a terminal time of 20 seconds; (d) SMC with a terminal time of 200 seconds.}\label{fig: SimA_energy flexibility}
\end{figure}

For D$^2$OC, the limited coverage of the reference distribution is shown in Fig. \ref{fig: SimA_energy flexibility}(a) compared to the longer operation time in Fig. \ref{fig: SimA_energy flexibility}(b). This is expected, as the shorter operation time results in fewer agent-points, leading to a less accurate representation of the reference density. 
These two simulation results well represent the benefit of D$^2$OC, energy flexibility, as the initially available energy acts as a physical constraint given by the system itself.
The D$^{2}$OC scheme adjusts the amount of the weight for each agent-point to account for the terminal time (or the remaining energy) of each agent. This characteristic of the proposed control enables the coverage of more critical areas in a limited amount of time, resulting in energy-flexible control.

On the other hand, by observing the trajectories originating from the initial locations in Figs. \ref{fig: SimA_energy flexibility}(c) and (d), it is evident that the trajectories in Fig. \ref{fig: SimA_energy flexibility}(c) overlap with those in Fig. \ref{fig: SimA_energy flexibility}(d). 
The control input for the SMC method is computed based on the Fourier coefficients as well as the past trajectories of the agents, which is independent of the agents' operation time. 
These overlaps in the trajectories result from the control input's independence of the terminal time, leading to insufficient coverage as demonstrated in Fig. \ref{fig: SimA_energy flexibility}(c). 
Alternatively, this behavior is interpreted as the future multi-agent ergodic trajectories simply continuing the current trajectories, since the operation time cannot be incorporated into the SMC method.

The SMC method could not cover some of the high-priority regions (top-left and bottom-right), whereas D$^2$OC achieved coverage within the same operation time, as shown in Figs. \ref{fig: SimA_energy flexibility}(a) and (c).
This distinction is critical in many applications such as search-and-rescue, environmental monitoring, smart farming, and infrastructure inspection, where the goal must be achieved with limited resources.

Further, the D$^2$OC scheme offers several advantages over the SMC method, including greater coverage efficiency and improved handling of symmetry-related issues. Fig. \ref{fig: SimA_energy flexibility}(c) presents the irrational behavior of the agents, characterized by movement along the horizontal and vertical axes from their initial locations. This irrational behavior is caused by the symmetry of the reference density function and the initial locations of the agents. On the contrary, the D$^2$OC scheme exhibits no issue concerning the symmetry.
Moreover, the SMC method inherently causes the agents to circulate between distinct high-priority regions, as shown in Fig. \ref{fig: SimA_energy flexibility}(d). This feature of the SMC method leads to inefficient coverage due to the agents traveling into less significant areas between high-priority regions.

To compare the D$^2$OC scheme with the SMC method quantitatively, the second simulation was conducted on a preconstructed reference distribution, shown in Figs. \ref{fig: Sim A-1 reference map}(a) and (b) as the contour plot of the reference distribution function and its sample-points representation, respectively. Table \ref{table: A parameter} shows the parameters used in this simulation.

\begin{figure}[hb]
    \centering
    \includegraphics[width=0.90\linewidth]{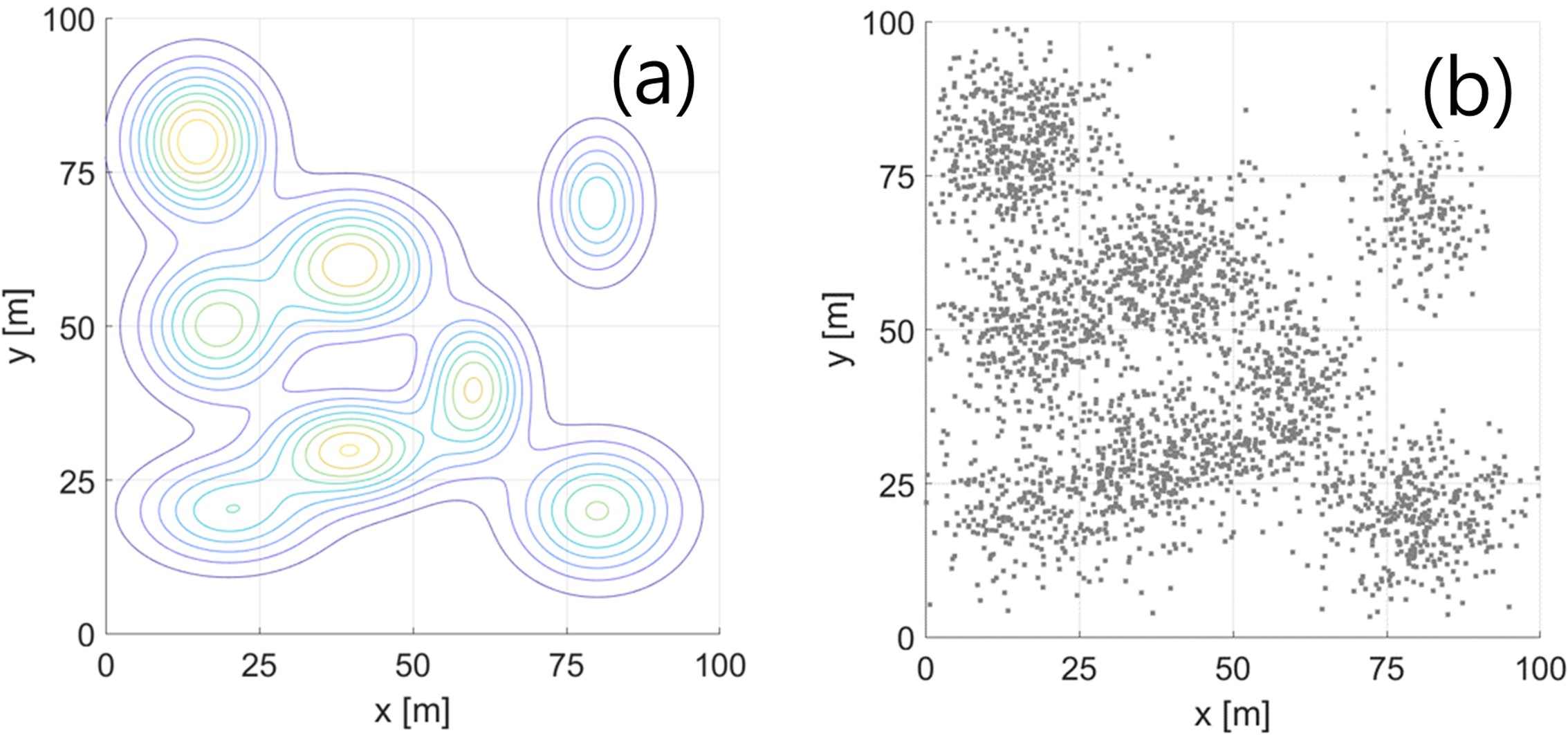}
    \caption{Reference distribution: (a) Contour plot of the density function; (b) Discrete sample-points distribution.}
    \label{fig: Sim A-1 reference map}
\end{figure}

\begin{table}[!hb]
\caption{Parameters of the Simulation for IV-A}
\label{table: A parameter}
\centering
\begin{tabular}{clll}
\hline
\multicolumn{2}{c}{Parameter Name or Meaning} & Symbol & Value \\ \hline

\multirow{6}{*}{\begin{tabular}{@{}c@{}}Environ.\\\&\\System\end{tabular}} 
& Domain size & -- & $100\,\text{m} \times 100\,\text{m}$ \\ \cline{2-4} 
& Number of agents & $L$ & 6 \\ \cline{2-4} 
& Input threshold & $u_{\text{max}}$ & $10\,\text{m/s}$ \\ \cline{2-4} 
& \begin{tabular}[c]{@{}l@{}}Total number of\\sample-points\end{tabular} & $N$ & 3000 \\ \cline{2-4}  
& Discrete-time interval & $\Delta t$ & $0.1\,\text{s}$ \\ \cline{2-4} 
& Terminal time & $t_f$ & \begin{tabular}[c]{@{}l@{}}100, 200, 300,\\400, 500\,\text{s}\end{tabular} \\ \hline

\multirow{2}{*}{D$^2$OC} 
& Penalty matrices & \multicolumn{2}{l}{$Q = \mathbf{0}$, $R = \mathbf{I}_2$} \\ \cline{2-4} 
& Horizon length & $T$ & 15 \\ \hline

\multirow{1}{*}{SMC} 
& Number of Fourier bases & $K$ & 10, 20, 30 \\ \hline

\end{tabular}
\end{table}

\begin{figure}[ht]
\centering
\includegraphics[width=0.90\linewidth]{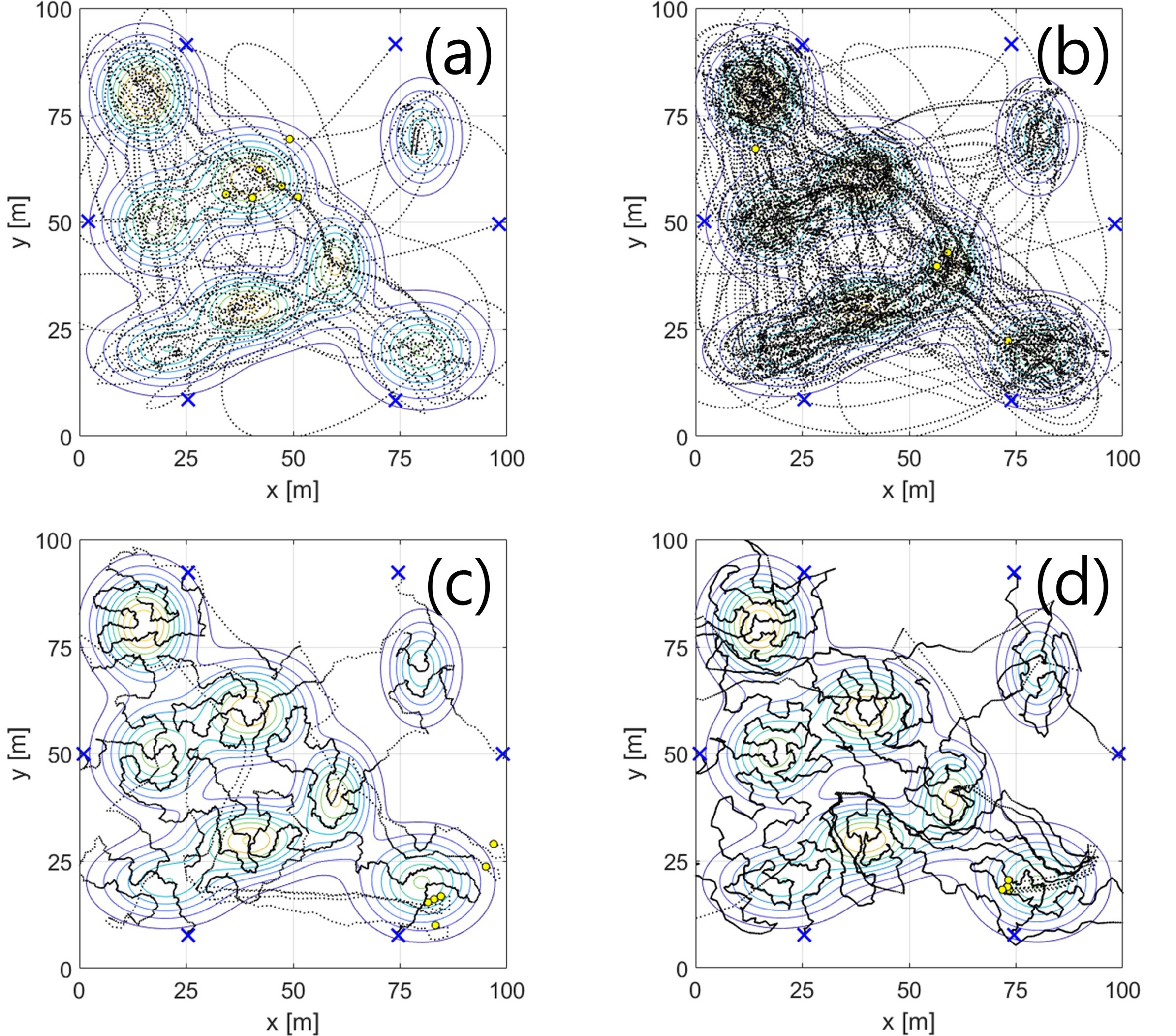}
	\caption{Agents' trajectories of SMC method (a-b) and D$^2$OC scheme (c-d) with varying terminal times. The initial and final locations of agents are indicated by blue-cross and yellow-circle marks, respectively: (a) SMC method (30 bases) with a terminal time of 100 seconds; (b) SMC method (30 bases) with a terminal time of 400 seconds; (c) D$^2$OC scheme with a terminal time of 100 seconds; (d) D$^2$OC scheme with a terminal time of 400 seconds.}\label{fig: SimA_trajectory_of_agents}
\end{figure}

\begin{figure}[ht]
    \centering
    \includegraphics[width=0.80\linewidth]{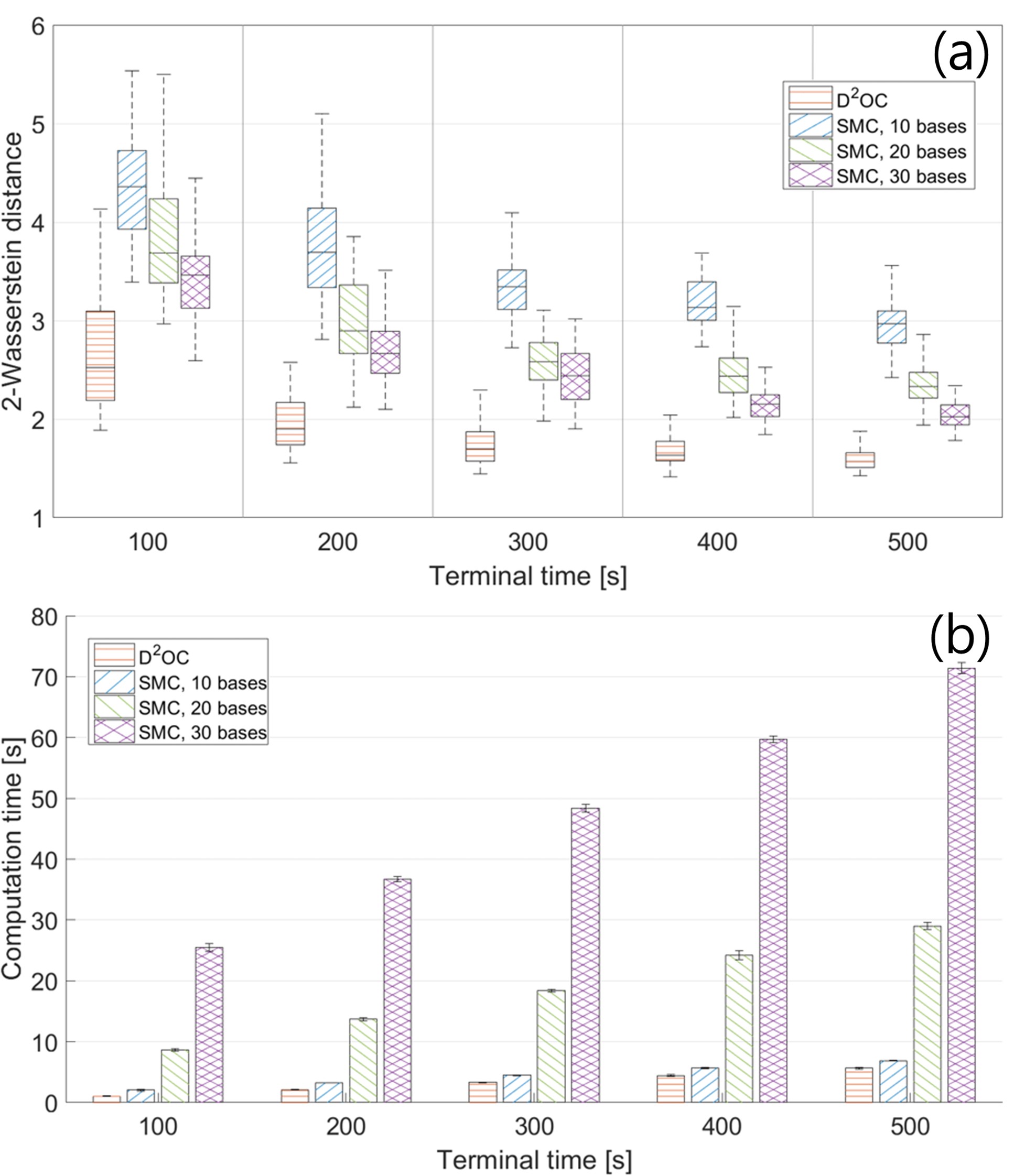}
	\caption{2-Wasserstein distance and Computation time of the SMC methods and the D$^{2}$OC scheme: (a) 2-Wasserstein distance; (b) Computation time.} \label{fig: SimA_comptime_wassdist}
\end{figure}

The simulations were performed with varying terminal times from 100 seconds to 500 seconds. Fig. \ref{fig: SimA_trajectory_of_agents} shows the discrete-time trajectories of agents in black dots when the terminal time is 100 seconds (Figs. \ref{fig: SimA_trajectory_of_agents}(a) and (c)) and 400 seconds (Figs. \ref{fig: SimA_trajectory_of_agents}(b) and (d)). The blue-cross and yellow-circle markers represent the initial and final locations of agents, respectively. These trajectory plots facilitate a qualitative assessment of coverage performance. It is evident from the SMC trajectories in Fig. \ref{fig: SimA_trajectory_of_agents}(b) that agents spent considerable time traveling low-priority areas, denoted by the low-value contours in the reference density function. This unnecessary travel is considerably reduced in D$^2$OC, as shown in Figs. \ref{fig: SimA_trajectory_of_agents}(c) and (d).

Figs. \ref{fig: SimA_comptime_wassdist}(a) and (b) present the Wasserstein distance and the computation time measured in both D$^2$OC and SMC methods. For the SMC method, three different numbers of Fourier bases---10, 20, and 30---are used. To reduce the undesired effect of agents' initial locations, 50 simulations were conducted for each case with random initial locations, and the results are presented as the mean values with their corresponding standard deviations. The same reference density distribution in Fig. \ref{fig: Sim A-1 reference map} was used.  
Fig. \ref{fig: SimA_comptime_wassdist}(a) shows the box plots for the 2-Wasserstein distance between two ensembles: the reference distribution and the distribution constructed by the agent-points. 
The 2-Wasserstein distance decreased as the terminal time increased for all cases. Moreover, the 2-Wasserstein distance of the D$^{2}$OC scheme was the smallest in all cases, indicating that the D$^2$OC scheme achieved a closer alignment of agent trajectories with the reference distribution than the SMC scheme.
 
Fig. \ref{fig: SimA_comptime_wassdist}(b) presents the computation time for each case, measured as the elapsed time between the start and end of each simulation. The D$^2$OC scheme also demonstrates superior performance compared to the SMC methods in terms of computation time. Although the computation times for the D$^{2}$OC method and the SMC method with 10 bases were similar, the D$^{2}$OC scheme achieved better performance in terms of 2-Wasserstein distance. This result clearly shows that, for a fixed computation time, the D$^{2}$OC scheme provides much better coverage of the reference distribution than the SMC method.

\subsection{Performance Comparison between D$^2$C and D$^2$OC}\label{sec: 4-B}
 To compare the performance of D$^2$OC with D$^2$C, a discrete-time linear quadrotor model \cite{sabatino2015quadrotor} is considered in this simulation. The state vector for this system at time $k$ is defined by 
 $\mathsf{x}^{k}:=[p_{x}^{k} \ dp_{x}^{k}\ \theta^{k} \ d\theta^{k} \  p_{y}^{k} \ dp_{y}^{k}\ \phi^{k} \  d\phi^{k}]^{\top}$, the input vector by $u^k=[\tau^k_{x} \ \tau^k_{y}]^{\top}$, and the output vector by $y^k=[p^k_{x} \ p^k_{y}]$. The state variables $p_{x}$ and $ p_{y}$ represent the $x$- and $y$-coordinate of the quadrotor, $\theta$ and $\phi$ are the roll and pitch angles, and $\tau_{x}$ and $\tau_{y}$ represent the torques applied along the $x$- and $y$-axes. The notation $d(\cdot)$ denotes the increment of the given variable
between consecutive time intervals. The small-angle assumption ($\theta,\phi \approx 0$) is made to neglect higher-order terms and approximate the transcendental terms to linear terms. External disturbances, such as wind, are assumed to be absent. The differential equations are linearized at the equilibrium point $\bar{x}^k = [p_{x}^{k}\ 0\ 0\ 0\ p_{y}^{k}\ 0\ 0\ 0]^\top$, so that they take the form described in \eqref{eqn: LTI system}. The detailed parameters and values used in the simulation are presented in Table \ref{table: B parameter}. It is noteworthy that the first and fifth diagonal entries of the penalty matrix $Q$ are set to zero. This is to set the penalties on the corresponding states, the 
$x$ and $y$-coordinates, to zero, thereby imposing no regulation on these states and allowing the agent to move freely.
  Simulations were conducted with varying terminal times as specified in Table \ref{table: B parameter}. For each condition, simulations were repeated 50 times while changing the agents' initial positions randomly. 

\begin{table}[!hb]
\caption{Parameters of the Simulation for IV-B}
\label{table: B parameter}
\centering
\begin{tabular}{clll}
\hline
\multicolumn{2}{c}{Parameter Name} & Symbol & Value \\ \hline

\multirow{8}{*}{\begin{tabular}{@{}c@{}}Environ.\\\&\\System\end{tabular}} 
& Domain size & -- & $100\,\text{m} \times 100\,\text{m}$ \\ \cline{2-4} 
& Number of quadrotors & $L$ & 6 \\ \cline{2-4} 
& Acceleration of gravity & $g$ & $9.81\,\text{m/s}^2$ \\ \cline{2-4} 
& Moment of inertia & $I_{xx}, I_{yy}$ & $0.0075\,\text{kg}\cdot\text{m}^2$ \\ \cline{2-4} 
& Discrete-time interval & $\Delta t$ & $0.1\,\text{s}$ \\ \cline{2-4} 
& \begin{tabular}[c]{@{}l@{}}Total number of\\sample-points\end{tabular} & $N$ & 3000 \\ \cline{2-4} 
& Communication range & $r_\text{comm}$ & $25\,\text{m}$ \\ \cline{2-4} 
& Terminal time & $t_f$ & \begin{tabular}[c]{@{}l@{}}60, 120, 180,\\240, 300\,\text{s}\end{tabular} \\ \hline

\multirow{2}{*}{D$^{2}$OC} 
& Penalty matrices & \multicolumn{2}{l}{\begin{tabular}[c]{@{}l@{}} $Q$ = $10^{-5}\times$diag(0,1,1,1,0,1,1,1)\\$R$ = $\mathbf{I}_2$ \end{tabular}} \\ \cline{2-4} 
& Horizon length & $T$ & 15 \\ \hline

\multirow{2}{*}{D$^{2}$C} 
& Penalty matrices (LQR) & \multicolumn{2}{l}{$Q = \mathbf{I}_8$, $R = \mathbf{I}_2$} \\ \cline{2-4} 
& \begin{tabular}[c]{@{}l@{}}Number of\\local sample-points\end{tabular} & $h$ & 2, 4, 5 \\ \hline

\end{tabular}
\end{table}

\begin{figure}[h]
    \centering
    \includegraphics[width=0.90\linewidth]{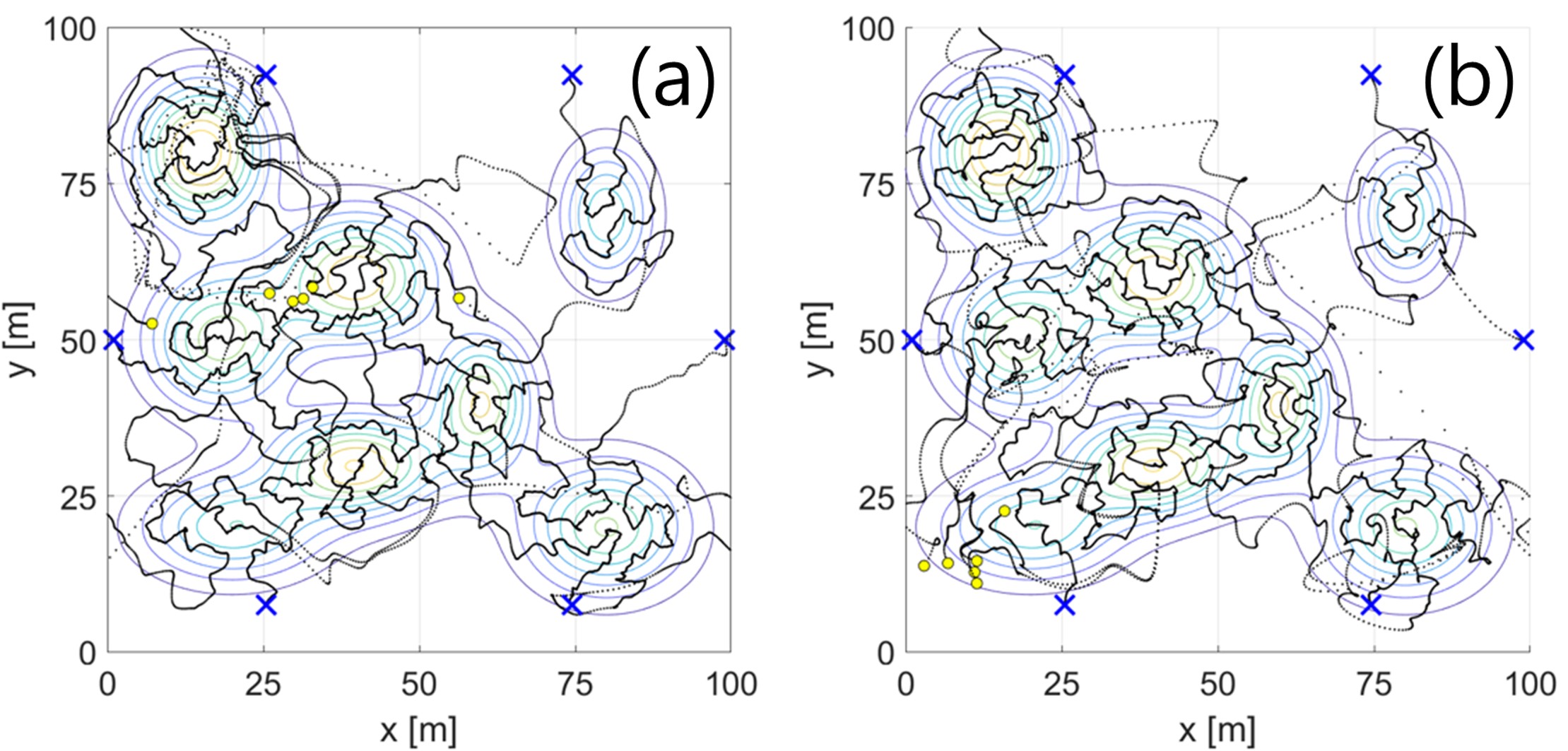}
    \caption{Agents' trajectories with D$^2$C and D$^2$OC schemes. The initial and final locations of agents are indicated by blue-cross and yellow-circle marks, respectively: (a) D$^{2}$C scheme (h = 5) with a terminal time of 240 seconds; (b) D$^{2}$OC scheme with a terminal time of 240 seconds.}\label{fig: SimB_trajectory_of_agents}
\end{figure}

  The trajectories of agents for D$^2$C and D$^2$OC schemes are shown in Fig. \ref{fig: SimB_trajectory_of_agents}. Notice that the D$^2$C scheme does not incorporate the agent's dynamics into the trajectory planning but only generates the next waypoint for the agent to visit. Thus, we applied an LQR controller for the quadrotor system to follow the generated waypoint via D$^2$C.
  One of the dominant parameters in D$^2$C is the number of local sample-points, $h$, which may significantly affect the performance of D$^2$C. A total of three different values for $h$ were tested, as shown in Table \ref{table: B parameter}. Determining the optimal value for $h$ often relies on heuristics and can be time-consuming, which is another drawback of D$^2$C. In contrast, D$^2$OC determines the local sample-points based on the OT theory constraint in \eqref{eqn: OT constraint}, eliminating the need to set $h$.

\begin{figure}[ht]
    \centering
    \includegraphics[width=0.90\linewidth]{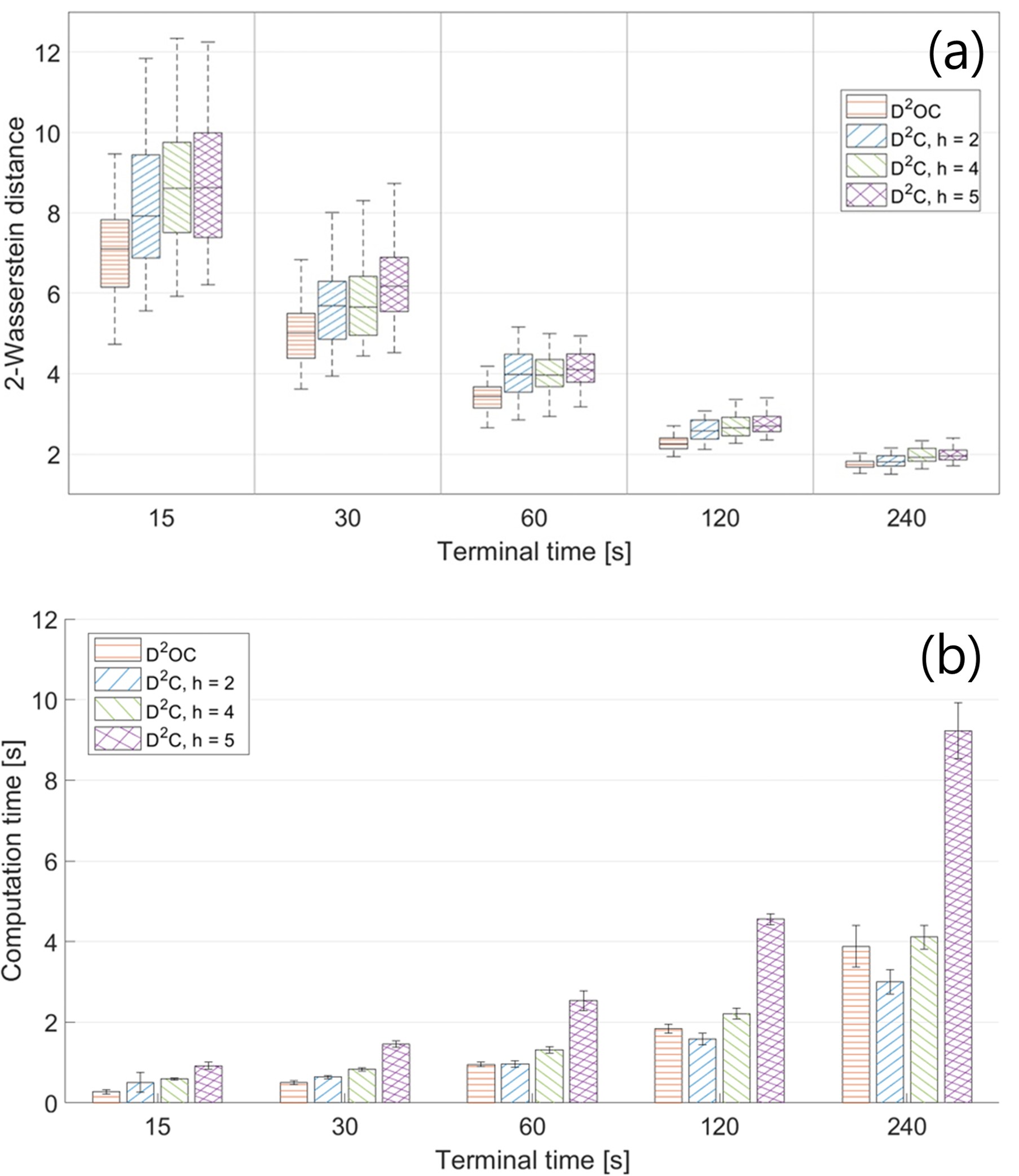}
	\caption{2-Wasserstein distance and the computation time of the D$^{2}$C and the D$^{2}$OC schemes (lower is better): (a) 2-Wasserstein distance; (b) the computation time.} \label{fig: SimB_comptime_wassdist}
\end{figure}

  Despite the fact that both trajectories appear similar at first glance in Fig. \ref{fig: SimB_trajectory_of_agents}, the quantitative results in Fig. \ref{fig: SimB_comptime_wassdist} demonstrate noticeable differences in performance.
  The 2-Wasserstein distance and computation times with varying terminal times are shown in Figs. \ref{fig: SimB_comptime_wassdist}(a) and (b). In Fig. \ref{fig: SimB_comptime_wassdist}(a), the 2-Wasserstein distance decreased as the terminal time increased, which indicates that the dissimilarity between the reference distribution and the distribution of agent-points is diminished. This is quite natural since the multi-agent system can better represent the given reference density distribution with more agent-points.
  The 2-Wasserstein distance of the D$^2$OC scheme was the smallest among the control schemes, which in turn implies that the D$^2$OC scheme outperformed the D$^2$C scheme in terms of the coverage of the reference distribution. 
  
 In Fig. \ref{fig: SimB_comptime_wassdist}(b), the average computation times for 50 simulations and $\pm$1 standard deviations are described. The result shows that the computation time of the D$^{2}$C scheme considerably increased with higher $h$. This is because in the D$^{2}$C scheme, the optimal path is identified by examining the costs of all possible permutations connecting $h$ local sample-points. The number of possible permutations increases as the number of local sample-points rises, thereby increasing the computational effort. On the other hand, the primary factors affecting the computation time in D$^2$OC are the length of the horizon and the dimension of the state vector. The use of the analytic solution in \eqref{eqn: optimal_u} enables D$^2$OC to achieve computational efficiency, leading to fast computation times.

\subsection{Performance Comparison in Weight Sharing}
In the decentralized control setup, the performance of the multi-agent coverage depends on how each agent shares the weight of the sample-point with nearby agents within the communication range.
The simulation was conducted to evaluate the performance difference between the original and proposed weight-sharing methods. 
The discrete-time linear quadrotor model described in Section \ref{sec: 4-B} was considered with most of the parameters identical to those in Table \ref{table: B parameter}. A total of 20 quadrotors were considered for two different scenarios in which the proposed Stages A and B were applied identically, while the weight-sharing methods differed---the original and proposed weight-sharing methods. This is for a fair performance comparison of the two weight-sharing methods while excluding other factors such as a different controller.

The penalty matrices for the D$^2$OC scheme were set to $Q=10^{-3}\times \text{diag}(0,\,10^{-2},\,1,\,1,\,0,\,10^{-2},\,1,\,1)$ and $R = \mathbf{I}_2$. 
For each scenario, the simulations were repeated 100 times with random initial locations for the quadrotors. Fig. \ref{fig: SimC_domain_traj}(a) presents the reference distribution, and (b) shows the resulting trajectories of agents. 
\begin{figure}
    \centering
\includegraphics[width=0.90\linewidth]{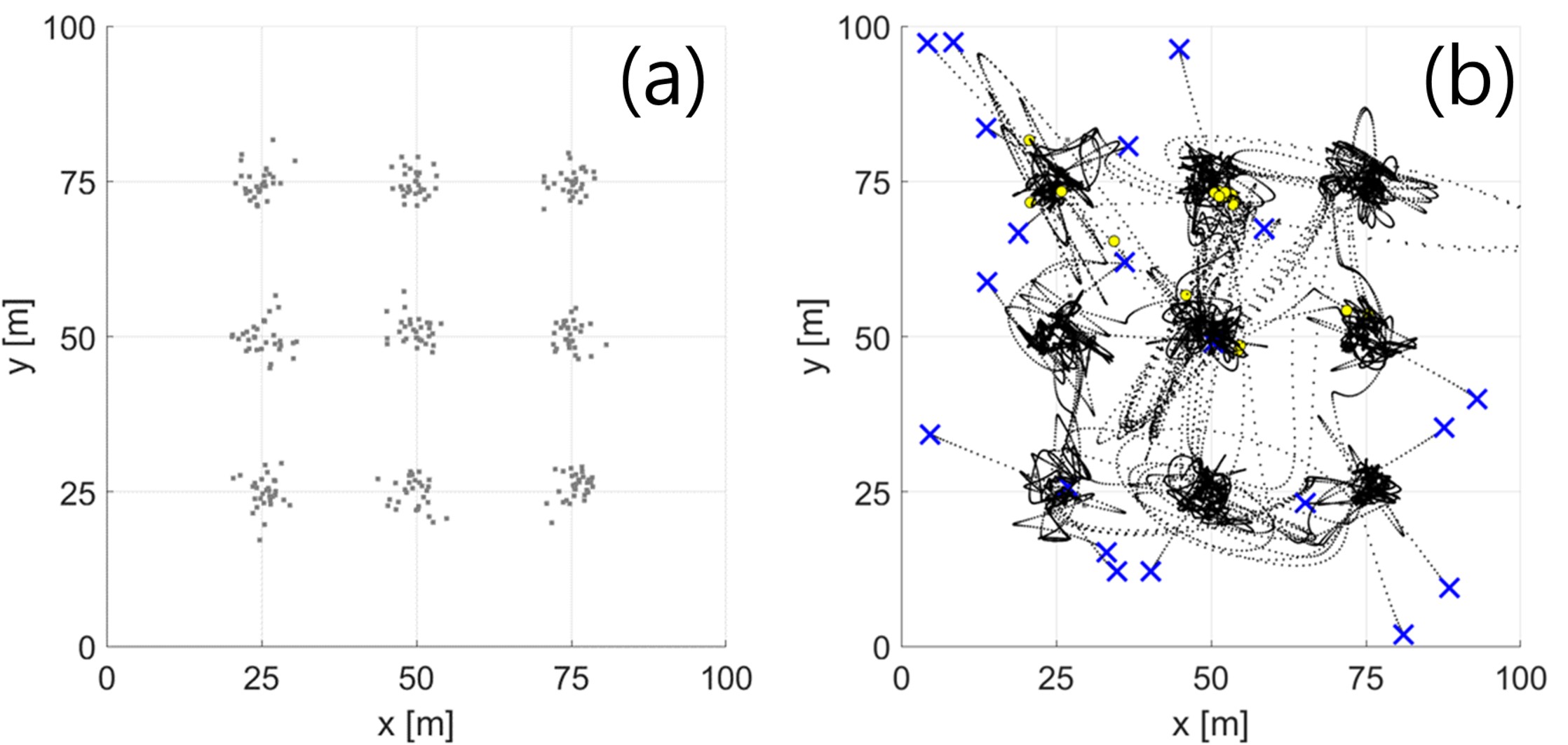}
    \caption{(a) Reference distribution; (b) Agents' trajectories using the D$^2$OC scheme that implements the proposed weight-sharing method. The initial and final locations of agents are indicated by blue-cross and yellow-circle marks, respectively. }\label{fig: SimC_domain_traj}
\end{figure}
Figs. \ref{fig: remaining 1}(a) and (b) show the temporal evolution of the average remaining weight for two communication ranges, $r_\text{comm}=1$ m and $r_\text{comm}=2$ m, where the average remaining weight is defined by
$
    \begin{aligned}
        &\text{Average Remaining Weight}
    \end{aligned}=\frac{1}{L}\left(\sum_{r=1,\,j=1}^{L,N} {}^{r}\beta_{j}^{k}\right)  \times 100\ [\%] \nonumber
$.

\begin{figure}[ht]
    \centering
    \includegraphics[width=0.90\linewidth]{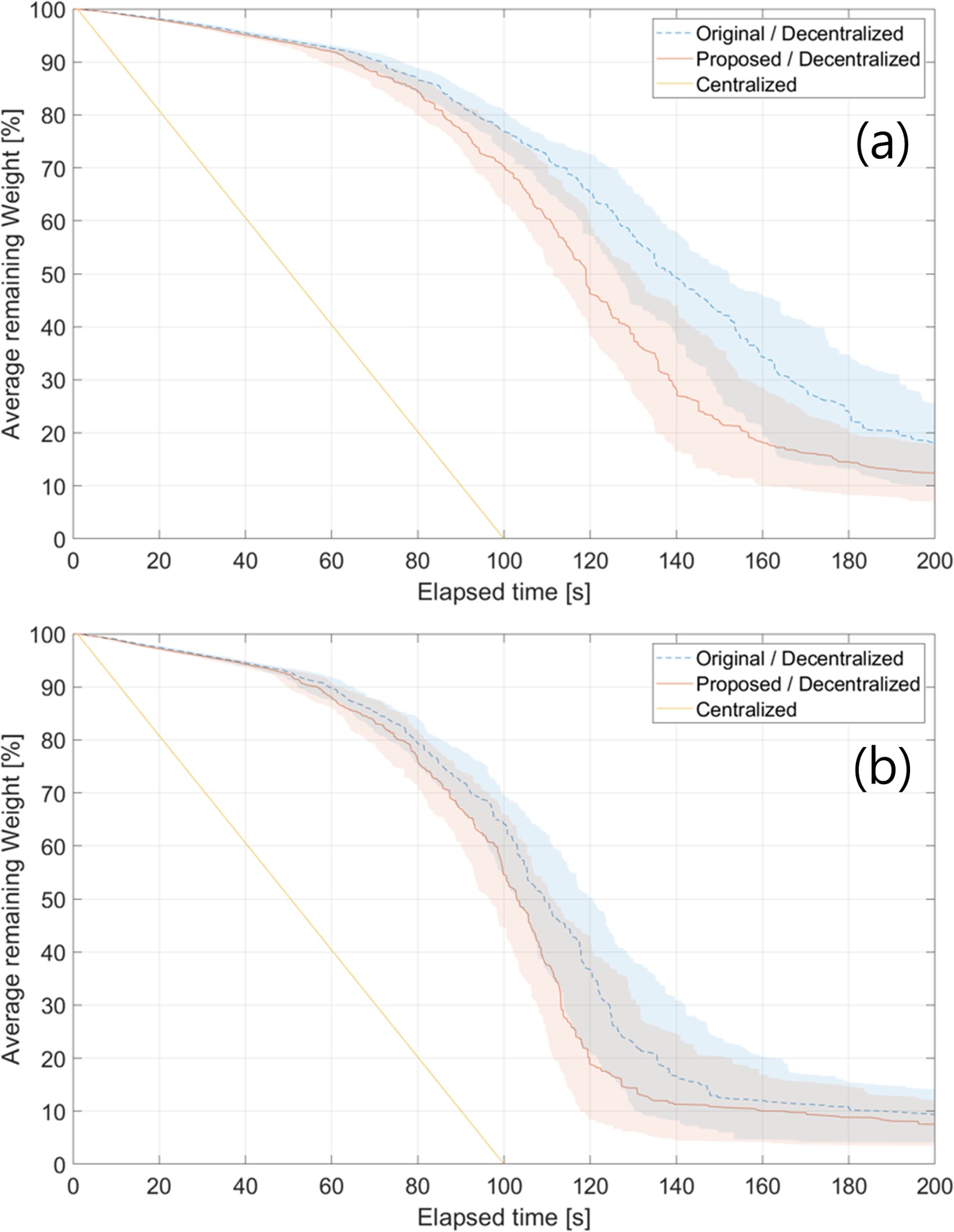}
    \caption{Temporal evolution of the remaining weight during the simulation for different communication ranges: (a) $r_\text{comm} = 1$ m; (b) $r_\text{comm} = 2$ m.}\label{fig: remaining 1}
\end{figure}

The median values for the original and proposed weight-sharing method are shown as a dashed blue line and a solid red line, respectively. The yellow solid line represents the average remaining weight for the centralized case, which can be considered the ideal scenario with all-to-all inter-agent communication. The shaded areas around the lines indicate the range of $\pm$1 standard deviation from the median values. 

The curve of the proposed weight-sharing method is closer to that of the centralized case than the original weight-sharing method, indicating higher efficiency. 
As shown in Proposition \ref{prop: general form minimum share}, the original weight-sharing method may accumulate the omission of the coverage progress, ${}^{r}\delta^{n}$, at each weight-sharing stage. This omission of the coverage progress leads to inaccuracies in the remaining weights of agents, causing more work redundancy.

In the case of $r_\text{comm}=2$ m, Fig. \ref{fig: remaining 1}(b) depicts that the curves of both sharing methods are more closely aligned to the curve of the centralized case, with a narrower gap between them. This observation is expected, given that the centralized case assumes an infinite communication range.

Another metric used to evaluate the coverage efficiency is work redundancy. This metric refers to the redundant efforts exerted by agents beyond what is required by the reference distribution. 
To measure how much redundant work
was done, a more formal notion of work redundancy is defined by $
    \text{Work Redundancy} = \left\{\left(\sum_{r,\,j=1,\,k=0}^{L,\,N,\,{}^{r}k_{f}} {}^{r}\gamma^{k}_{j}\right)-1\right\} \times 100\ [\%], \nonumber
$
where ${}^{r}k_{f}$ represents the terminal time of agent $r$.

The work redundancy of the agents is calculated and plotted in Fig. \ref{fig: Surplus work}. In the centralized case (ideal case), the work redundancy is zero. As shown in Fig. \ref{fig: Surplus work} for both cases ($r_\text{comm}=1$ m and $r_\text{comm}=2$ m), the proposed weight-sharing method yields lower work redundancy than the original method, thereby improving the overall efficiency of the agents' effort.

\begin{figure}[h]
    \centering
    \includegraphics[width=0.80\linewidth]{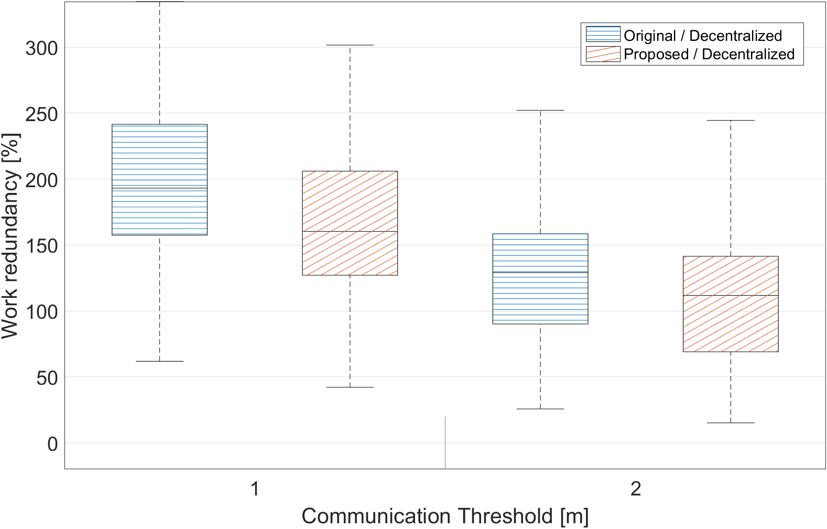}
    \caption{Work redundancy of the agents due to decentralized control.}\label{fig: Surplus work}
\end{figure}

\section{Conclusion}\label{sec: conclusion}
This paper investigated the density-driven optimal control scheme for efficient and collaborative multi-agent coverage problems. Due to constraints in a given mission, such as the finite number of agents, limited operation time for each agent, and a communication range limit, uniform coverage may not be feasible, especially for a spacious domain. The proposed control scheme can be a practical solution to tackle these issues by incorporating the importance or priority of the domain as a reference density distribution. The proposed density-driven optimal controller is derived from the Lagrangian associated with the Wasserstein distance. For the LTI system, both the global optimality and the existence of the optimal control input are guaranteed with the analytic solution. For decentralized control, an efficient weight-sharing rule is proposed, which is more efficient than the existing method by reducing work redundancy between multiple agents. Moreover, the proposed D$^2$OC can provide energy-flexible control by reflecting initially available energy into the plan. Various simulation results were presented to compare the performance between the existing and proposed works. It is verified that the proposed density-driven optimal control scheme outperforms other existing methods both in efficiency, measured by the Wasserstein distance, and in computation time.

\bibliographystyle{IEEEtran}%unsrt, ieeetr
\bibliography{references}
\end{document}